\documentclass[reqno,12pt]{amsart}
\usepackage[english]{babel}
\usepackage{amsmath}
\usepackage{amsaddr}

\usepackage{amssymb}
\usepackage{amsthm}
\usepackage[scr=boondoxo]{mathalfa}
\usepackage{graphicx}
\usepackage{setspace}
\usepackage{epstopdf}

\usepackage{subfigure}
\usepackage{color}
\usepackage[dvipsnames]{xcolor}
\usepackage{enumitem}

\newcommand{\bpr}{\begin{trivlist} \item[]{\bf Proof. }}
\newcommand{\epr}{\hspace*{\fill} $\qed$\end{trivlist}}

\newcommand{\be}{\begin{eqnarray}}
\newcommand{\ee}{\end{eqnarray}}
\newcommand{\ba}{\begin{align}}
\newcommand{\ea}{\end{align}}
\newcommand{\bi}{\begin{itemize}}
\newcommand{\ei}{\end{itemize}}

\newcommand{\secref}[1]{Section~\ref{sec:#1}}

\newcommand{\seclab}[1]{\label{sec:#1}}
\newcommand{\eqlab}[1]{\label{eq:#1}}
\renewcommand{\eqref}[1]{(\ref{eq:#1})}

\newcommand{\figref}[1]{Fig.~\ref{fig:#1}}
\newcommand{\figlab}[1]{\label{fig:#1}}

\newcommand{\lemmaref}[1]{Lemma~\ref{lemma:#1}}
\newcommand{\lemmalab}[1]{\label{lemma:#1}}

\newcommand{\remlab}[1]{\label{remark:#1}}

\newcommand{\thmref}[1]{Theorem~\ref{theorem:#1}}
\newcommand{\thmlab}[1]{\label{theorem:#1}}

\newcommand{\appref}[1]{Appendix~\ref{app:#1}}

\newcommand{\applab}[1]{\label{app:#1}}

\newcommand{\R}{\mathbb R}

\newtheorem{theorem}{Theorem}[section]

\newtheorem{lemma}[theorem]{Lemma}

\newtheorem{remark}[theorem]{Remark}

\numberwithin{equation}{section}
\newcommand\KUK[1]{{\color{black}{#1}}}
\newcommand\rspp[1]{{\color{black}{#1}}}
\newcommand\rsp[1]{{\color{black}{#1}}}

\begin{document}
\title[A dynamical systems approach to WKB-methods]{A dynamical systems approach to WKB-methods: The eigenvalue problem for a single well potential}

\author{K. Uldall Kristiansen*}
\address{Department of Applied Mathematics and Computer Science, 
Technical University of Denmark, 
2800 Kgs. Lyngby, 
Denmark }
\email{krkri@dtu.dk}
\author{P. Szmolyan}
\address{Institute of Analysis and Scientific Computing, TU Wien, Vienna, Austria
}
\date\today
\maketitle

\vspace* {-2em}

 \begin{abstract}
 In this paper, we revisit the eigenvalue problem of the one-dimensional Schr{\"o}dinger equation for smooth single well potentials. In particular, we provide a new interpretation of the Bohr-Sommerfeld quantization formula.  A novel aspect of our results, which are based on recent work of the authors on the turning point problem based upon dynamical systems methods, is that 
we cover all eigenvalues $E\in [0,\mathcal O(1)]$ and show that the Bohr-Sommerfeld quantitization formula approximates all of these eigenvalues (in a sense that is made precise). At the same time, we provide rigorous smoothness statements of the eigenvalues as functions of $\epsilon$. We find that whereas the small eigenvalues $E=\mathcal O(\epsilon)$ are smooth functions of $\epsilon$, the large ones $E=\mathcal O(1)$ are smooth functions of $n\epsilon \in[c_1,c_2],\,0<c_1<c_2<\infty$, and $0\le \epsilon^{1/3}\ll 1$; here $n\in \mathbb N_0$ is the index of the eigenvalues. 
 \end{abstract}
 \section{Introduction}
 In this paper, we reconsider the quantization of energies $E$ as described by the one-dimensional Schr{\"o}dinger equation
 \begin{align}
  \epsilon^2 \ddot x &=\left(V(t)-E\right)x,\eqlab{eq0}
 \end{align}
in the semi-classical limit $\epsilon \rightarrow 0$. 
Here $x(t)$ is the wave function, $t \in \R$ a spatial variable,  $V(t)$ the potential 
and  $E$  the energy. The eigenvalue problem consists of {determining the values of $E$ for which \KUK{$L^2$-integrable} solutions $x:\mathbb R\rightarrow \mathbb R$ exists}. For $E < V(t)$ solutions are exponentially growing/decaying, for $E >V(t)$ solutions are oscillatory. This is obvious when $V$ is constant, but carries over to the time-dependent case whenever $0<\epsilon\ll 1$.
\KUK{This behaviour changes at \textit{turning points}} $t=t_*(E)$ where $V(t_*(E)) =E$. In the corresponding classical dynamics these points are the points where
the velocity of the corresponding particle changes its sign, hence the name turning point.

{It is well-known (see e.g. \cite[Section 2.3]{berizin}) \KUK{that the spectrum of the Schr{\"o}dinger operator $x\mapsto -\epsilon^2 \ddot x +V(t) x\in L^2$} -- under mild growth conditions of $V(t)\rightarrow \infty$ for $t\rightarrow \pm \infty$ -- is discrete: $\{E_n(\epsilon)\}_{n\in \mathbb N_0}$ with $E_n(\epsilon)\rightarrow \infty$ for $n\rightarrow \infty$ for all $\epsilon>0$. In fact, (due to K. Friedrichs \cite{friedrichs1934a}) the statement also holds true in arbitrary dimensions $d\in \mathbb N$, where one replaces the second order derivative $\ddot{()}$ by the Laplacian $\Delta=\sum_{i=1}^d \frac{\partial^2}{\partial t_i^2}$ with $t=(t_1,\ldots,t_d)\in \mathbb R^d$, see also \cite{berizin,reed1972a}.}
\KUK{In this paper, we will consider \eqref{eq0} in the case of a single well potential: $V'(t)=0\,\Leftrightarrow\, t=0$ \rspp{and focus on} the asymptotics of $E_n(\epsilon)$ with respect to $\epsilon\to 0$.}  
\rsp{In particular,} we consider $V\in C^\infty$ and suppose that
\begin{enumerate}[label=(\Alph*)]
 \item \label{Vquadratic} $V(0)=0$, $V'(t)=0 \Leftrightarrow t=0$ and $V''(0)=2$, so that $t=0$ is a global minimum with value $V(0)=0$,
\end{enumerate} 
as well as
 $\lim_{t\rightarrow \pm \infty} V(t)=\infty$ (see the separate technical assumptions of the paper \ref{smoothness} and \ref{growth} below). (Notice that although $V''(0)=2$ looks restrictive, it is (through simple scalings) without loss of generality once we suppose that $V''(0)>0$.)  In this case, there are two turning points $t_-(E)<0<t_+(E)$ for every $E>0$. We define
\begin{align}
 J(E) := \frac{1}{\pi}\int_{t_-(E)}^{t_+(E)} \sqrt{E-V(t)} dt,\quad E>0.\eqlab{Jdefn}
\end{align}
 It is then common ``\rsp{folkore}'', see \cite[Section 15.2]{hall2013} for a recent mathematics-oriented reference, that eigenvalues of \eqref{eq0} can be approximated by the Bohr-Sommerfeld \KUK{quantization} formula:
\begin{align}
 J(E_n) = \left(n+\frac12\right) \epsilon\quad n\in \mathbb N.\eqlab{quant}
\end{align}
The equivalent form
\begin{align}
\int_{t_-(E_n)}^{t_+(E_n)} \sqrt{E_n-V(t)}dt=\pi J(E_n) =\pi\left(n+\frac12\right)\epsilon,\nonumber
\end{align}
using \eqref{Jdefn} in the first equality and \eqref{quant} in the second,
is perhaps more familiar. The fraction $1/2$ appearing on the right hand side is known as the Maslov correction, see \cite{keller}.

The quantization condition \eqref{quant} can be determined in a formal way based upon the WKB-method (see e.g. \cite{bender1978,fedoryuk1993,olver1974a}) for $E\ge c>0$. In particular, WKB gives rise to an exponentially decaying solution for $t \to -\infty$:
\begin{align*}
 \rspp{x_-^{\rm{exp}}}(t) \approx \frac{1}{ \sqrt[4]{V(t)-E} } \exp   \left ( - \frac{1}{\epsilon}\int \sqrt{V(t)-E}\,dt\right ),
\end{align*}
within the ``classically forbidden region'' $(-\infty, t_-(E))$. This solution 
can be connected (formally) to an oscillatory WKB-solution
\begin{align*}
x_{-}^{\rm{osc}}(t) \approx
 \operatorname{Re}\left(\frac{1}{ \sqrt[4]{E-V(t)} } \exp   \left ( \pm \frac{i}{\epsilon}\int \sqrt{E-V(t)}\,dt\right )u\right),
\end{align*}
for some $u\in \mathbb C$, 
within the ``classically allowed region'' $(t_-,t_+)$ through the Airy-function:
\begin{align}
 x(t) = \operatorname{Ai} \left(-\epsilon^{-2/3} (-V'(t_{-}(E)))^{1/3} (t-t_-(E))\right),\eqlab{airysol}
\end{align}
since this function solves
the local model of \eqref{eq0}
\begin{align*}
 \epsilon^2 \ddot x & = V'(t_-(E)) (t-t_-(E)) x,
\end{align*}
near $t=t_-(E)$. Similarly,   the exponentially decaying WKB-solution \rsp{$x_+^{\rm{exp}}$} for $t \to \infty$ within the ``classically forbidden region'' $(t_+,\infty)$, can be connected to 
an oscillatory solution   $x_{+}^{\rm{osc}} $  in the ``classically allowed region'' $(t_-,t_+)$. The quantization condition is then obtained (formally) from the requirement 
 $x_{-}^{\rm{osc}} (t) = x_{+}^{\rm{osc}}(t) $ for all  $t \in  (t_-,t_+)$.

If we continue to suppose that $E\ge c>0$, then the left hand side of \eqref{quant} satisfies $J(E)\ge \bar c>0$ for some $\bar c>0$. Consequently, \eqref{quant} implies that $n=\mathcal O(\epsilon^{-1})$ for $\epsilon\rightarrow 0$. This makes the interpretation of \eqref{quant}, as a formula for the approximation of eigenvalues as $\epsilon\to 0$, nontrivial. Moreover, from this perspective there is no justification for \eqref{quant} approximating all (bounded) eigenvalues. 
 
 The following is also well-known (\rsp{see e.g. \cite{bender1978,landau2013quantum,simon1984a}}): For every $n\in \mathbb N_0$ there is an eigenvalue
 \begin{align}
  E_n\approx 2\left(n+\frac12\right)\epsilon,\eqlab{Esmall}
 \end{align}
 for $\epsilon\rightarrow 0$. \eqref{Esmall} is clearly not uniform with respect to $n$ and these eigenvalues are therefore for $n=0,\ldots,n_0$, $n_0\in \mathbb N$ fixed, the \textit{small} (or low-lying) eigenvalues, each being $E_n=\mathcal O(\epsilon)$ with respect to $\epsilon\rightarrow 0$. The corresponding eigenfunctions have $n$ zeros with $n$ being the number in \eqref{Esmall}.
 
 The two sets of eigenvalues, \textit{large} of order $\mathcal O(1)$ (given by \eqref{quant}) and the \textit{small} ones of order $\mathcal O(\epsilon)$ (given by \eqref{Esmall}) do not overlap as $\epsilon\rightarrow 0$. Consequently, we have a set of eigenvalues in between, which we shall refer to as the \textit{intermediate} eigenvalues. To the best of our knowledge, these eigenvalues and the lack of overlap between \eqref{quant} and \eqref{Esmall} appear to be somewhat overlooked; \rsp{we refer to \cite{weinstein1985a,weinstein1987a} for a treatment of such gaps in the context of periodic potentials}. A possible explanation for this might be that \eqref{quant} with $V(t)= t^2$ gives the leading order expression \eqref{Esmall} \KUK{for all $n\in \mathbb N_0$}. 
 

%

 Although the quantization of eigenvalues and the Bohr-Sommerfeld approximation is covered \KUK{in  several} classical references, see e.g. the textbooks \cite{fedoryuk1993,quantum}, we have found the reference \cite{yafaev2011a} to provide the most precise mathematical description of the problem in the context of assumption \ref{Vquadratic} and under the assumption that $V$ is smooth (not analytic). Following a similar approach to the one used in \cite{fedoryuk1993}, \cite[Theorem 2.5]{yafaev2011a} provides an expansion of solutions across the turning points. These solutions are then concatenated using the associated Wronskian in order to set up an equation for the associated eigenvalue problem. It is shown that this leads to the Bohr-Sommerfeld equation up to remainder terms of order $\mathcal O(\epsilon^2)$, see \cite[Eq. (4.5)]{yafaev2011a}. However, we believe that the subtleties related to $n\rightarrow \infty$ and the smoothness of eigenvalues with respect to $\epsilon$ are not addressed. 
 
 
 
 In \cite{uldall2024a}, the present authors provided a detailed description of turning points using an alternative approach based upon dynamical systems theory (including singular perturbation theory, normal form theory and the blowup method).  We believe that this approach has the advantage of being more systematic \KUK{and sheds light on the above mentioned smoothness issues}. In particular, we are in the process of applying this approach in related problems with elliptic and hyperbolic transitions. 
 In \cite{uldall2024a}, we were also able to address smoothness with respect to $\epsilon$. In line with \eqref{airysol} and asymptotics of $\operatorname{Ai}$, we found that solutions across the turning point are smooth functions of $\epsilon^{1/3}$; \KUK{in fact, more precisely they are smooth functions of $\epsilon^{2/3}$ and $\epsilon$ but we will stick with the former simpler formulation throughout}. We have not seen this elsewhere in the literature. This is potentially due to the fact that we do the matching in a different way (using blowup). We will discuss this fact further below. 
 
 The present paper can be seen as a companion paper to \cite{uldall2024a}. Here we again pursue modern dynamical systems based techniques to address the eigenvalues of \eqref{eq0}, including the smoothness properties with respect to $\epsilon$ and the dependency on $n$. Our results will (in contrast to \cite{yafaev2011a}) lead to an equation for the eigenvalues which depends smoothly on $n\epsilon $ and $\epsilon^{1/3}$. In particular, we find that this equation is a $\mathcal O(\epsilon^{5/3})$-perturbation (and not a $\mathcal O(\epsilon^2)$-perturbation as in \cite{yafaev2011a}) of the Bohr-Sommerfeld approximation. For these two separate results to match, the $\epsilon^{5/3}$-term in our expansion needs to be absent, but we have neither performed a detailed calculation to verify this (this is not a trivial task) nor have we found a direct explanation either. We hope to shed further light on the discrepancies in future work.

 
%


\subsection{Main results}
In this paper, we provide a rigorous description of the eigenvalues of \eqref{eq0} for a single well potential, assuming that \ref{Vquadratic} and the following conditions all hold true:
\begin{enumerate}[resume*]
\item \label{smoothness} $V:\mathbb R\rightarrow \mathbb R$ is $C^\infty$.
 \item \label{growth} 
 $V(t)\rightarrow \infty$ as $t\rightarrow\pm \infty$ such that the function
 \begin{align}
  t\mapsto V(t)^{-\frac14} \frac{d^2}{dt^2}\left(V(t)^{-\frac14}\right),\eqlab{phifunc}
 \end{align}
which is defined for all $t\ne 0$ by assumption \ref{Vquadratic}, 
 belongs to $L^1((-\infty,-1])\cap L^1([1,\infty))$. 
%
%
%
%
\end{enumerate}

\rspp{Assumptions \ref{smoothness} and \ref{growth} imply that \eqref{eq0} for all $0<\epsilon\ll 1$, and all $E\in \mathbb R$, have non-trivial solutions $x_\pm^{\rm{exp}}:\mathbb R\rightarrow \mathbb R$  that have exponential decay as $t\rightarrow \pm \infty$, respectively. This follows from first applying the Liouville transformation described in \cite[Theorem 1.1, p. 190]{olver1974a}, see also \cite[Eq. (106), p. 191]{olver1974a}, and subsequently using  Proposition 8.1 from \cite[p. 92]{coddington1955a}. (Notice that we simply require that the $\phi$-function in \cite[Eq. (106), p.191]{olver1974a} is integrable with respect to $\xi$. Upon a change of coordinates, this leads to the statement regarding \eqref{phifunc}, see also \cite[Proposition 4.4]{berizin}.) Assumption \ref{growth} also relates to \cite[Assumption 2.1]{yafaev2011a} (although our formulation is slightly weaker). }

\begin{remark}
 As also noted in \cite{yafaev2011a} in the context of \cite[Assumption 2.1]{yafaev2011a}, assumption \ref{growth} is \rsp{satisfied} for all potentials with algebraic growth where $V(t)=\mathcal O(t^{M})$, $V'(t)=\mathcal O(t^{M-1})+\mathcal O(1)$ and $V''(t)=\mathcal O(t^{M-2})+\mathcal O(1)$ for $t\rightarrow \pm \infty$ for some  $M\in \mathbb N$. Indeed, in this case the function \eqref{phifunc} is $\mathcal O(t^{-2})$ for all $M\in \mathbb N$ as $t\rightarrow \pm\infty$. This follows from a simple calculation. The same is true for logarithmic and exponential growth of $V$. On the other hand, the function $V(t)=b \sin (t^3)+t^2$ with $b>0$ small enough is an example of a single well potential (satisfying \ref{Vquadratic} and \ref{smoothness}) that does not satisfy \ref{growth}; in this case \eqref{phifunc} \rspp{becomes} $\frac{9}{4} b \sin(t^3) t + \mathcal O(t^{-1})$ for $t\rightarrow \pm \infty$.
\end{remark}

Our approach is based upon dynamical systems theory, and \rspp{we will therefore write \eqref{eq0} as the first order system}:
\begin{equation}\eqlab{xyt}
\begin{aligned}
 \dot x &=y,\\
 \dot y &=(V(t)-E) x,\\
 \dot t&=\epsilon.
 \end{aligned}
\end{equation}
\rspp{Within this (geometric) viewpoint, the existence of $x_\pm^{\rm{exp}}$ implies that there are well-defined stable and unstable manifolds $W^{s}(\epsilon,E)$, $W^u(\epsilon,E)$ of the line $(0,0,t)$, $t\in \mathbb R$, for \eqref{xyt} for $t\gg 1$ and $t\ll -1$, respectively.  $W^i(\epsilon,E)$, $i=s,u$, can each be extended by the flow for all $t\in \mathbb R$.} Obviously, by the linearity of the problem, $W^s(\epsilon,E)$ and $W^u(\epsilon,E)$ are $C^\infty$ line bundles (cf. assumption \ref{smoothness}). Geometrically $E$ being an eigenvalue, with an associated eigenfunction $x:\mathbb R\rightarrow \mathbb R$ with $x(t)\rightarrow 0$ (exponentially) as $t\rightarrow \pm \infty$, means that the manifolds $W^s(\epsilon,E)$ and $W^u(\epsilon,E)$ coincide. 

 \begin{figure}[!ht] 
\begin{center}
{\includegraphics[width=.78\textwidth]{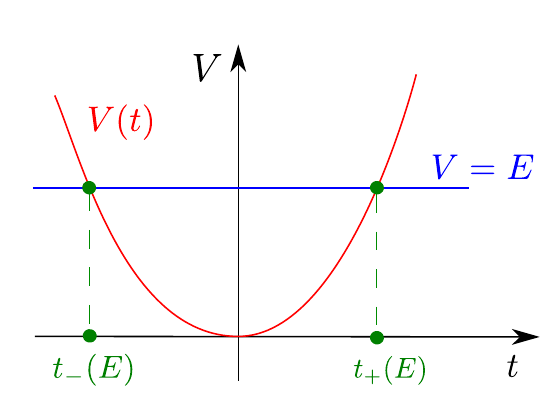}}
\end{center}
 \caption{A single well potential $V(t)$ having a global minimum at $t=0$. We also illustrate the turning points $t_\pm(E)$, $E>0$, where $U(t,E)=E-V(t)$ changes sign.  }
\figlab{potentialV}
\end{figure}

Consider \eqref{xyt} with $\epsilon=0$ and $t<t_-(E)$. Then the linearization \rsp{about} $(0,0,t)$ produces two real and nonzero eigenvalues $\pm \lambda(t)$ with $\lambda(t)=\sqrt{V(t)-E}$. In this way, we obtain an unstable manifold $W^u(0,E)$ for $\epsilon=0$, being the line bundle of the unstable spaces of $(0,0,t)$ associated with the eigenvalue $\lambda(t)$, and by standard hyperbolic theory, $W^u(\epsilon,E)$ is a smooth perturbation within compact subsets of $t<t_-(E)$. $W^s(\epsilon,E)$ is similarly a perturbation of a stable manifold \rsp{$W^s(0,E)$} for \eqref{xyt} with $\epsilon=0$ within compact subsets of $t>t_+(E)$. However, at $t=t_\pm(E)$ we have that $\lambda(t)=0$ and the hyperbolic theory offers no control of these manifolds within the elliptic regime  $t\in (t_-(E),t_+(E))$ \rsp{where  $\lambda(t)\in i\mathbb R$. In this region, solutions of \eqref{xyt} are rapidly oscillating for $0<\epsilon\ll 1$, see \figref{xyt_new}.}

\begin{figure}[!ht] 
\begin{center}
{\includegraphics[width=.99\textwidth]{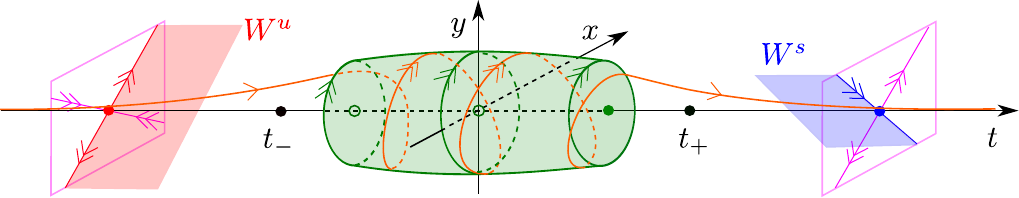}}
\end{center}
 \caption{\rsp{Illustration of our geometric viewpoint: For $t<t_-$ and $t>t_+$, we show the unstable manifold $W^u$ (red) and the stable manifold $W^s$ (blue), respectively. In these regions, these manifolds are approximated to leading order by the unstable and stables spaces, respectively, of \eqref{xyt} for $\epsilon=0$. In between $t\in (t_-,t_+)$, we have fast oscillations as indicated in green (using orbits of \eqref{xyt} for $\epsilon=0$). We emphasize that $E$ is an eigenvalue when the extensions of $W^u$ and $W^s$ coincide. An associated eigenfunction for $0<\epsilon\ll 1$ is illustrated in orange.   } }
\figlab{xyt_new}
\end{figure}

When reference to $\epsilon$ and $E$ is not important, we write $W^i(\epsilon,E)=W^i$, $i=s,u$, for simplicity.

\begin{lemma}\lemmalab{Jprop}
Suppose that the assumptions \ref{Vquadratic}, \ref{smoothness} and \ref{growth} all hold. Then the following holds.
\begin{enumerate}
\item \label{tpmE} The functions $t_\pm:\mathbb R_+\rightarrow \mathbb R$ defined by 
\begin{align}\eqlab{tpmEeqns}
E-V(t_-(E))=E-V(t_+(E)) = 0,
\end{align}
with $t_-(E)<0<t_+(E)$, are $C^\infty$-smooth functions and satisfy \begin{align}
 t_\pm (E)\rightarrow \pm \infty,\eqlab{tasymp}
\end{align}
for $E\rightarrow \infty$.
\item \label{Jprop2}
$J:\mathbb R_+\rightarrow \mathbb R_+$ defined by 
\eqref{Jdefn}
is a $C^\infty$ diffeomorphism.
\item \label{Jprop3} $J$ extends smoothly to $E=0$, with $J(0)=0$ and $J'(0)=\frac12$.
\end{enumerate}
\end{lemma}
This result can be found in \cite[Chapter 2.5.2.1]{fedoryuk1993} (but without proof). For completeness, we include a proof in \appref{t0E}. 

\bigskip 
The following theorem \rsp{provides a rigorous description of the asymptotics of $E=\mathcal O(1)$ eigenvalues. In particular, we believe it} sheds \rsp{a new} light on how the Bohr-Sommerfeld approximation \eqref{quant} may be understood. 

\begin{theorem}\thmlab{01eigenvalues}
Fix a \KUK{compact interval $D\subset \mathbb R_+$} and 
let $\widetilde D\subsetneq J(D)$, recall the definition of $J$ in \eqref{Jdefn}. Then there exists an $\epsilon_0>0$ sufficiently small and two \KUK{$C^\infty$}-smooth functions
\begin{align*}
 \overline J(\cdot,\cdot,\pm 1):\widetilde D\times [0,\epsilon_0^{1/3})\rightarrow \mathbb R,
\end{align*}
such that the following holds for all $0<\epsilon< \epsilon_0$:
For any $n\in \mathbb N$ with $n \epsilon\in \widetilde D$ there is an eigenvalue $E_n\in D$ of \eqref{eq0} given by
\begin{align*}
 J(E_n)= \left(n+\frac12\right)\epsilon + \epsilon^{5/3} \overline J( n\epsilon,\epsilon^{1/3},(-1)^n).
\end{align*}
In other words, $E_n=E_n(n\epsilon,\epsilon^{1/3})$ is a \KUK{$C^\infty$}-smooth function of $n\epsilon$ and $\epsilon^{1/3}$.

The spacing between adjacent eigenvalues $E_{n+1}-E_n$ is $\mathcal O(\epsilon)$ and each of the open intervals $(E_{n},E_{n+1})\subset \mathbb R_+$ contain no additional eigenvalues.

\end{theorem}

Essentially the result says that \eqref{quant} holds up to $\mathcal O(\epsilon^{5/3})$ for $n\epsilon\in \widetilde D$ with the remainder being smooth with respect to $n\epsilon$ and $\epsilon^{1/3}$. \rsp{Notice that the eigenvalues are $\mathcal O(1)$ since the domain $D\subset \mathbb R$ is compact and fixed.} As already mentioned, the smoothness of  $\overline J$ with respect to $\epsilon^{1/3}$ rather than just $\epsilon$, \rspp{could be the consequence of our treatment of the turning point, see \cite{uldall2024a} and \secref{review} below}. 
In any case, the order of the remainder $\mathcal O(\epsilon^{5/3})$ is not optimal according to \cite{yafaev2011a}. Here it is stated, see \cite[Theorem 4.1]{yafaev2011a}, that the Bohr-Sommerfeld approximation is valid up to $\mathcal O(\epsilon^2)$-remainder terms. 

\bigskip 

We also provide a new dynamical systems based proof of the following result \rspp{on the low-lying eigenvalues $E=\mathcal O(\epsilon)$}.
\begin{theorem}\thmlab{smalleigenvalues}
Fix \KUK{$n_0\in \mathbb N$} and let $N_0:=\{0,1,\ldots,n_0\}\subset \mathbb N_0$. 
Then there exists an \KUK{$\epsilon_0=\epsilon_0(n_0)>0$} such that the following holds: For every $n\in N_0$ there exists a \KUK{$C^\infty$}-smooth function $e_n:[0,\epsilon_0)\rightarrow \mathbb R$ so that 
\begin{align*}
 E_n(\epsilon) = (2n+1)\epsilon+\epsilon^2 e_n(\epsilon),
\end{align*}
is an eigenvalue of \eqref{eq0}. The open intervals $(E_{n},E_{n+1})\subset \mathbb R_+$, \mbox{$n\in N_0\setminus \{n_0\}$}, contain no additional eigenvalues. 
\end{theorem}

For the proof of \thmref{smalleigenvalues} (see further details in \secref{smalleigenvalues}) we show that $e:=\epsilon^{-1} E\in D_1$, with $D_1$ a fixed but large compact interval, are roots of a smooth Melnikov (or Evans) function 
\begin{equation}\eqlab{Delta1}
\begin{aligned}
\Delta_1 &: D_1\times [0,\sqrt{\epsilon_0})\rightarrow \mathbb R,\\
&(e,\sqrt{\epsilon}) \mapsto \Delta_1(e,\sqrt{\epsilon}),
\end{aligned}
\end{equation}
satisfying:
\begin{align}
 \Delta_1(2n+1,0)=0 \quad \text{and}\quad  \frac{\partial }{\partial e} \Delta_1(2n+1,0)\ne 0,\eqlab{Delta1cond}
\end{align}
for all $n\in \mathbb N$. Using the implicit function theorem, we then obtain eigenvalues $E=\epsilon e_n(\sqrt{\epsilon})$, with $e_n(0)=2n+1$, as smooth functions of $\sqrt{\epsilon}$. However, using an indirect argument (based upon a time \rspp{reversal} symmetry of a scaled system), we show that $e_n$ extends smoothly to an even function of $\sqrt{\epsilon}$, and therefore in turn conclude that the eigenvalues are in fact smooth functions of $\epsilon$ as claimed. 

\bigskip 

Finally, we prove the following regarding the intermediate eigenvalues:
 \begin{theorem}\thmlab{intermediateEigenvalues}
 Consider \eqref{eq0}. Then for $\epsilon_0>0$, $\xi_0>0$, $r_{20}>0$, and $c_2>0$ all sufficiently small and $c_1>0$ sufficiently large, we have the following for any $\epsilon\in (0,\epsilon_0)$: There exist two \KUK{$C^\infty$}-smooth functions
  \begin{align*}
 \overline J_2(\cdot,\cdot,\pm 1):[0,\xi_0^{1/3})\times [0,r_{20}]\rightarrow \mathbb R,
 \end{align*}
satisfying
\begin{align}
 \overline J_2((2m+1)^{-1/3},0,(-1)^m)=0,\eqlab{barJ2prop}
 \end{align}
 for all $m\in \mathbb N_0$ with  $(2m+1)^{-1}\le \xi_0$,
 such that 
 $E\in [c_1 \epsilon,c_2]$ is an eigenvalue if and only if there is an $n\in \mathbb N$ such that 
 \begin{align}
  J(E) = \left(n+\frac12\right)\epsilon + \epsilon^{3/2} (E^{-1}\epsilon)^{\frac{1}{6}} \overline J_2((E^{-1}\epsilon)^{1/3},E^{\frac12},(-1)^n).\eqlab{JEintermediate}
 \end{align}
\end{theorem}
We see that upon truncating \eqref{JEintermediate} to leading order, we obtain \eqref{quant}. Moreover, we will see that \eqref{barJ2prop} is a consequence of the equality in \eqref{Delta1cond}.
Upon using \lemmaref{Jprop} item \ref{Jprop3}, it is also possible to solve \eqref{JEintermediate} for $E\in [c_1 \epsilon,c_2]$  as a function of $n$ and $\epsilon$. However, we have not found a good way of representing the smoothness properties of such solutions. In any case, these solutions of \eqref{JEintermediate} can be chosen to overlap with those in \thmref{smalleigenvalues} (by taking $n_0>0$ large enough) and \thmref{01eigenvalues} (by taking $D$ (and $\widetilde D$) large enough). Consequently, in this way, we cover all bounded eigenvalues for $0<\epsilon\ll 1$, and as a corollary that \eqref{quant} approximates all of these. 

We believe that \eqref{JEintermediate}, as an implicit equation for $E$, provides a good representation of how the different regimes $E=\mathcal O(\epsilon)$ and $E=\mathcal O(1)$  in (\thmref{01eigenvalues} and \thmref{smalleigenvalues}, respectively) are connected. Indeed, for $E=\epsilon e$ with $e\ge c_1$, \eqref{JEintermediate} becomes
\begin{align*}
 J(\epsilon e) = \left(n+\frac12\right) \epsilon  + \epsilon^{3/2} e^{-\frac{1}{6}} \overline J_2(\rspp{e^{-1/3}},\epsilon^{\frac12} e^{\frac12},(-1)^n),
\end{align*}
with the right hand side being a smooth function of $r_2=\sqrt{\epsilon}$.
This is in agreement with \eqref{Delta1}, see also \lemmaref{Delta1} below. Moreover by dividing this equation by $\epsilon$ on both sides, and using \lemmaref{Jprop} item \ref{Jprop3} and \eqref{barJ2prop}, we obtain
\begin{align*}
 e = 2n+1,
\end{align*}
for $\epsilon\rightarrow 0$, in agreement with \thmref{smalleigenvalues}.
On the other hand, for $E=\mathcal O(1)$ we have 
\begin{align*}
  J(E) =  \left(n+\frac12\right) \epsilon + \epsilon^{5/3} E^{-\frac{1}{6}} \overline J_2(E^{-1/3}\epsilon^{1/3},E^{\frac12},(-1)^n),
\end{align*}
with the right hand side being a smooth function of $n\epsilon$ and $\epsilon^{1/3}$. This is in agreement with \thmref{01eigenvalues}, see also \lemmaref{solveDeltaEq0} below. 
\begin{remark}\remlab{ref1}
 \rsp{The change of smoothness with respect to $\epsilon$ (in the present context from smooth with respect to $\epsilon$ to smooth with respect to $\epsilon^{1/3}$) is a common feature in asymptotic expansions associated with singular perturbation problems with turning points. In general, the occurrence of fractional powers of $\epsilon$ in asymptotic expansions can be directly understood through the weights of the blow-up transformation used in the analysis of the turning point. For further background, we refer to \cite{freddy,krupa_extending_2001,gils2005a}, see also \cite{uldall2024a} for the use of GSPT and blow-up in the context of the simple turning point problem.}
\end{remark}

\subsection{\rsp{Structure of the paper}}
\rsp{In \secref{review}, we first review the results of \cite{uldall2024a} on the tracking of the unstable manifold across a simple turning point. We summarize the results in \thmref{DeltaThm} which forms the foundations for proving \thmref{01eigenvalues} in \secref{01eigenvalues}. {It is important to emphasize that we assume $C^\infty$-smooth potentials ($C^k$-smoothness is also possible, see \secref{discussion}), recall \ref{smoothness}, and therefore we do not rely upon complex integration paths valid only for analytic potentials $V\in C^\omega$, see \cite{fedoryuk1993}.} We prove \thmref{smalleigenvalues} and \thmref{intermediateEigenvalues} in \secref{smalleigenvalues} and \secref{intermediateEigenvalues}, respectively. The proofs of these results are organized around a blowup in parameter space of $(\epsilon,E)=(0,0)$.  %
We conclude the paper in \secref{discussion}. }


\section{A dynamical systems approach to simple turning points}\seclab{review}
 The functions $t_\mp(E)$ are simple roots of 
 \begin{align}
 U(t,E):=E-V(t),
 \end{align} 
 for $E>0$, recall \lemmaref{Jprop} item (\ref{tpmE}), and therefore give rise to two turning point problems:
 \begin{equation}\eqlab{airy1}
 \begin{aligned}
  \dot x &=y,\\
  \dot y&=-\mu(t,E)x,\\
  \dot t&=\epsilon,
 \end{aligned}
\end{equation}
 with 
 \begin{align}\eqlab{mudef}\mu(t,E)=U(t_-(E)+t,E)\quad \mbox{and}\quad \mu(t,E)=U(t_+(E)-t,E),
 \end{align}
 satisfying $\mu(0,E)=0$, $\frac{\partial}{\partial t}\mu(0,E)>0$. In \cite{uldall2024a}, we studied turning point problems of the form \eqref{airy1} and showed the following (see \cite[Theorem 3.2]{uldall2024a}):
  \begin{theorem}\thmlab{localWu}
  Consider \eqref{airy1} and suppose that 
  \begin{align}
   \mu:I\times D\rightarrow \mathbb R
  \end{align}
with $I$ a neighborhood of $t=0$ and $D\subset \mathbb R_+$ as a compact interval, is a $C^\infty$-function, that $$\mu(0,E)=0, \quad \frac{\partial }{\partial t}\mu(0,E)>0,\quad \mbox{for all}\quad  E\in D,$$ and consider any $M\in \mathbb N$. Finally, let $W^u(\epsilon,E)$ denote the unstable manifold of $(0,0,t)$ for $t\le -c<0$ for all $0\le \epsilon\ll 1$ and $E\in D$. 
 Then there exist an $\epsilon_0=\epsilon_0(M)>0$ and a $\nu>0$ both small enough, such that for all $\epsilon\in (0,\epsilon_0)$:
\begin{align}
W^u(\epsilon,E)\cap \{t=\nu\} = \left\{ 
\begin{pmatrix}
              x\\
              y
             \end{pmatrix}
\in \operatorname{span}\begin{pmatrix}
 X(\epsilon^{1/3},E) \\
  -{\sqrt{\mu (\nu,E)}}Y(\epsilon^{1/3},E)
\end{pmatrix}\right\},\eqlab{WuExpr2}
\end{align}
where 
\begin{align*}
 X(\epsilon^{1/3},E) &= \cos\left(\frac{1}{\epsilon }\int_{0}^{\nu} \sqrt{\mu(s,E)}ds -\frac{\pi}{4}+ \epsilon^{2/3} \phi_1(\epsilon^{1/3},E)\right),\\
 Y(\epsilon^{1/3},E) &= 
 (1+\epsilon^{2/3}\rho(\epsilon^{1/3},E))\sin\left(\frac{1}{\epsilon }\int_{0}^{\nu} \sqrt{\mu(s,E)}ds -\frac{\pi}{4}+ \epsilon^{2/3} \phi_2(\epsilon^{1/3},E)\right)
\end{align*}
with 
$\rho,\phi_1,\phi_2:[0,\epsilon_0^{1/3})\times D\rightarrow \mathbb R$ all $C^M$-smooth.
\end{theorem}
The proof of this theorem in \cite{uldall2024a} is based upon applying a blowup of the degenerate points $(x,0,0)$ for \eqref{airy1} for $\epsilon=0$:
\begin{align}
 (r,(\bar y,\bar t,\bar \epsilon))\mapsto \begin{cases}
                                           y=r\bar y,\\
                                           t\,=r^2 \bar t,\\
                                           \epsilon\, =r^3 \bar \epsilon,
                                          \end{cases}\eqlab{blowup0}
\end{align}
where $r\ge 0$, $(\bar y,\bar t,\bar \epsilon)\in S^2\subset \mathbb R^3$. 
By scaling $t$ and $\epsilon$ it is without loss of generality to consider $\frac{\partial }{\partial t}\mu(0,E)=1$ for all $E\in D$. Then $\lim_{t\rightarrow 0} t^{-1} \mu(t,E)=1$ and the Airy equation:
\begin{align}\eqlab{airyeq}
\epsilon^2 x''(t) = - tx,
\end{align}
provides an approximate model for $t=\mathcal O(\epsilon^{2/3})$. This is made precise by the blowup \eqref{blowup0}, insofar that the $\bar \epsilon=1$-chart:
\begin{align*}
 (y_2,t_2,r_2)\mapsto \begin{cases}
                       y =r_2 y_2,\\
                       t \,= r_2^2 t_2,\\
            \epsilon \,=r_2^3,
                      \end{cases}
\end{align*}
with chart-specific coordinates $(y_2,t_2,r_2)$, $r_2=\epsilon^{1/3}$, gives
\begin{align*}
 \dot x &=y_2,\\
 \dot y_2 &=-t_2 \rsp{x},\\
 \dot t_2 &=1,
\end{align*}
for $r_2\rightarrow 0$. This system
is obviously equivalent to \eqref{airyeq}. Hence $x(t)=\operatorname{Ai}(-\epsilon^{-2/3} t)$, $y=-\epsilon^{1/3} \operatorname{Ai}'(-\epsilon^{-2/3} t)$ provides an accurate tracking of the unstable manifold $W^u$ within this regime. 

For $t\ge c \epsilon^{2/3}$, we use a diagonalization procedure. Basically, we apply a transformation of the form
\begin{align*}
 \begin{pmatrix}
  x\\
  y 
 \end{pmatrix} = \begin{pmatrix}
 f(t,\epsilon) &\overline f(t,\epsilon)\\
 \lambda(t) & \overline \lambda(t)\end{pmatrix} \begin{pmatrix}
u\\
v\end{pmatrix},
\end{align*}
where $\lambda(t)=i\sqrt{\mu(t,E)}$. (In the following, we suppress the dependency on $E$ for simplicity).
A simple calculation shows that the resulting equations for $u$ and $v$ are diagonalized provided that $f$ satisfies the following equation:
\begin{align}
 \epsilon \frac{\partial f}{\partial t}(t,\epsilon) = \lambda(t) (1-f(t,\epsilon)^2) +\epsilon \lambda(t)^{-1} \lambda'(t) f(t,\epsilon),\eqlab{feqn}
\end{align}
or as a first order system:
\begin{equation}\eqlab{feqn2}
\begin{aligned}
 \dot f &=\lambda(t) (1-f^2) +\epsilon \lambda(t)^{-1} \lambda'(t) f,\\
 \dot t&=\epsilon.
\end{aligned}
\end{equation}
Consider first $t \ge c>0$. 
 Then $\mu$ is uniformly bounded away from zero and $f=1$ is a normally elliptic critical manifold of \eqref{feqn2} for $\epsilon=0$. Indeed the linearization of \eqref{feqn2} for $\epsilon=0$ around $(1,t)$, $t\ge c>0$, has a single nonzero eigenvalue $-2\lambda(t)\in i\mathbb R\setminus \{0\}$. Therefore if $\mu$ is real analytic, then it follows from \cite{de2020a} that there exists a (local) solution $f(t,\epsilon)$ of \eqref{feqn}, analytic in $t$ and Gevrey-1 in $\epsilon$. If $\mu$ is smooth (as in our case, see \eqref{mudef}), then there are quasi-solutions:
\begin{lemma}\lemmalab{uvN}
Fix any $N$, suppose that $\mu$ is smooth and consider $t\in I_+$ so that $\mu(t)\ge c>0$. Then there exists a smooth function $f_N$, being polynomial of degree $N$ with respect to $\epsilon$, such that the transformation $(u,v,t)\mapsto (x,y)$ defined by
  \begin{align}\eqlab{xyuvN}
  \begin{pmatrix}
   x\\
   y
  \end{pmatrix}&=
\begin{pmatrix}
   f_N(t,\epsilon) & \overline f_N(t,\epsilon)\\
   \lambda(t) & -\lambda(t)
  \end{pmatrix}\begin{pmatrix}
  u\\
  v
  \end{pmatrix},
 \end{align}
brings \eqref{airy1} into the following near-diagonal form
 \begin{equation}\eqlab{uvtN}
\begin{aligned}
 \dot u &=\nu_N(t,\epsilon) u+\mathcal O(\epsilon^{N+1})v,\\
 \dot v&=\mathcal O(\epsilon^{N+1})\rsp{u}+\overline{\nu}_N(t,\epsilon)v,\\
 \dot t&=\epsilon.
\end{aligned}
\end{equation}
Here 
\begin{align}
\nu_N(t,\epsilon) &= \lambda(t) -\frac12 \lambda(t)^{-1} \lambda'(t)\epsilon +\epsilon^2 T_{2,N}(t,\epsilon),\eqlab{nupm}
\end{align}
for some smooth $T_{2,N}$. The remainder terms in \eqref{uvtN} are also smooth functions of $t\in I_+$ and $\epsilon$.
\end{lemma}

\begin{proof}
%
Instead of solving \eqref{feqn} exactly, we look for ``quasi-solutions'' defined in the following sense: Write the equation \eqref{feqn} as $F(f,t,\epsilon)=0$. Then $f_N(t,\epsilon)$ smooth is a ``quasi-solution'' of order $\mathcal O(\epsilon^{N+1})$ if $F(f_N(t,\epsilon),t,\epsilon)=\mathcal O(\epsilon^{N+1})$ for $N\in \mathbb N$ uniformly in $t\in I_+$. It is standard that such quasi-solutions can be obtained as Taylor-polynomials $f_N(t,\epsilon)=\sum_{n=0}^{N} R_n(t)\epsilon^n$, see e.g. \cite{de2020a,kristiansenwulff}, with $R_n$ recursively starting from $R_0(t)=1$ in the present case. In fact, a simple calculation shows that $R_n$ is given by
\begin{align}
 R_n &=-\frac12 \sum_{l=1}^{n-1} R_l R_{n-l}+\frac12 \lambda^{-1} R_{n-1}'+\frac12 \lambda^{-2} \lambda' R_{n-1},\eqlab{Rneqn}
\end{align}
for $n\ge 1$. 
Consequently, we find that for each fixed $N\in \mathbb N$ there is transformation (which is polynomial with respect to $\epsilon$) so that \eqref{uvtN} holds with smooth off-diagonal remainder terms of order $\mathcal O(\epsilon^{N+1})$. 
\end{proof}
By integrating the near-diagonal system \eqref{uvtN}, we obtain the following.
\begin{lemma}
 Consider \eqref{uvtN} and an initial \rsp{condition} $(u(t_0),v(t_0))$ at $t=t_0$. Suppose for simplicity that $N\ge 4$. Then 
 \begin{equation}\eqlab{uvsol}
 \begin{aligned}
  u(t) = \exp\left(\frac{1}{\epsilon}\int_{t_0}^t \nu_N(s,\epsilon)ds\right)\left((1+\mathcal O(\epsilon^N))u(t_0)+\mathcal O(\epsilon^N)v(t_0)\right),\\
  v(t) = \exp\left(\frac{1}{\epsilon}\int_{t_0}^t \overline \nu_N(s,\epsilon)ds\right)\left(\mathcal O(\epsilon^N)u(t_0)+(1+\mathcal O(\epsilon^N))v(t_0)\right).
 \end{aligned}
 \end{equation}
 Here each of the $\mathcal O(\epsilon^N)$ remainder-terms  are $C^{\lfloor \frac{N}{2}\rfloor-1}$-smooth jointly in $t$, $\epsilon\ge 0$, and $E$, with the order of each of the terms changing as follows:
 \begin{align}
  \frac{\partial^{k+l+m}}{\partial t^{k} \partial \epsilon^{l} \partial E^{m}} \mathcal O(\epsilon^N) = \mathcal O(\epsilon^{N-(k+2l+m)}),\eqlab{derivatives}
 \end{align}
 for all $0\le k+l+m\le \lfloor \frac{N}{2}\rfloor-1$.
\end{lemma}
\begin{proof}
 We set $u(t)=Q(t,t_0,\epsilon) \tilde u(t)$,  $v(t)=\overline Q(t,t_0,\epsilon) \tilde v(t)$ where 
 \begin{align*}
  Q(t,t_0,\epsilon)=\exp\left(\frac{1}{\epsilon}\int_{t_0}^t \nu_N(s,\epsilon)ds\right).
 \end{align*}
Then 
\begin{equation}\eqlab{tildeuveqns}
\begin{aligned}
 \frac{d\tilde u}{dt}&= \mathcal O(\epsilon^N) Q(t,t_0,\epsilon)^{-1} \overline Q(t,t_0,\epsilon) \tilde v(t),\\
 \frac{d\tilde v}{dt}&= \mathcal O(\epsilon^N) \overline Q(t,t_0,\epsilon)^{-1}  Q(t,t_0,\epsilon) \tilde u(t).
\end{aligned}
\end{equation}
Since $\vert Q^{-1} \overline Q\vert=1$, we can integrate these equations using uniform bounds to obtain
\begin{equation}\eqlab{tildeuvsols}
\begin{aligned}
 \tilde u(t) &=(1+\mathcal O(\epsilon^N))\tilde u(t_0)+\mathcal O(\epsilon^N)\tilde v(t_0),\\
 \tilde v(t) &=\mathcal O(\epsilon^N)\tilde u(t_0)+(1+\mathcal O(\epsilon^N))\tilde v(t_0).
\end{aligned}
\end{equation}
\eqref{derivatives} is obtained by differentiating \eqref{tildeuveqns}. The $\mathcal O(\epsilon^N)$-terms are regular and the result therefore follows by differentiating $Q$ and integrating the variational equations using uniform bounds (as in \eqref{tildeuvsols}). For further details we refer to \cite{uldall2024a}.
\end{proof}
Now returning to the turning point problem and $t\ge c\epsilon^{2/3}$, the results of \cite{uldall2024a} show -- by working in the $\bar t=1$ chart associated with \eqref{blowup0}:
  \begin{align*}
 (y_1,r_1,\epsilon_1)\mapsto \begin{cases}
                       y =r_1 y_1,\\
                       t \,= r_1^2,\\
            \epsilon \,=r_1^3\epsilon_1,
                      \end{cases}
\end{align*}
with chart specific coordinates $(y_1,r_1,\epsilon_1)$ --
that the quasi-diagonalization  of \lemmaref{uvN} can be extended to $t\ge c\epsilon^{2/3}$ for $c>0$ large enough. \KUK{For this we (also) use abstract results on complex saddle-nodes, see \cite[Lemma 4.8]{uldall2024a}}.  More precisely, there exists a quasi-solution $f(t,\epsilon)$ of \eqref{feqn} up to \rsp{remainder terms} that are $\mathcal O(\epsilon^{2N/3})$ for $t\ge c\epsilon^{2/3}$. Upon composing \eqref{uvsol} with the tracking of $W^u$ by the Airy-function for $t=\mathcal O(\epsilon^{2/3})$, we obtain the following representation 
\begin{equation}\eqlab{uvWu}
\begin{aligned}
 u(\nu) &=e^{-i\frac{\pi}{4}} e^{\frac{i}{\epsilon}\int_0^{\nu} \sqrt{\mu(s)} ds} (1+\mathcal O(\epsilon^{2/3})),\quad 
 v(\nu)=\overline u(\nu),
\end{aligned}
\end{equation}
of $W^u$ in the $(u,v)$-space at $t=\nu>0$ for $\nu>0$ small enough. This leads to \eqref{WuExpr2} upon carefully studying the remainder. In fact, this expression can be extended to any $\nu>0$ for which $\mu(t)>0$ for all $t\in (0,\nu]$. For this, we can just use \eqref{uvsol} with \eqref{uvWu} as initial conditions with $t_0=\nu$. Importantly, by this extension we see the functions $\rho$, $\phi_1$ and $\phi_2$ in the expression for \eqref{WuExpr2} are each $C^M$-smooth jointly in $\epsilon^{1/3}$, $E$ and $\nu$ (upon taking $N$ large enough). 

%


%

We therefore obtain the following corollary from \cite[Theorem 3.2]{uldall2024a} on $W^u(\epsilon,E)$ and $W^s(\epsilon,E)$. 
\begin{theorem}\thmlab{DeltaThm}
 Consider \eqref{xyt} and suppose that \ref{Vquadratic}, \ref{smoothness} and \ref{growth} all hold and fix \KUK{$D\subset \mathbb R_+$} as a compact interval. Let $\nu>0$ be small enough, $E\mapsto t_0(E)$ be a \KUK{$C^\infty$}-smooth function such that $t_-(E)+\nu<t_0(E)<t_+(E)-\nu$ for all $E\in D$ and consider the unstable and stable manifolds $W^u(\epsilon,E)$ and $W^s(\epsilon,E)$ of \eqref{xyt} for $t\rightarrow -\infty$ and $t\rightarrow \infty$, respectively. 
 Then there is an $\epsilon_0>0$ such that the following holds for any $\epsilon\in (0,\epsilon_0)$: 
 \begin{align*}
  W^u(\epsilon,E) \cap \{t=t_0(E)\} &=\left\{\begin{pmatrix}x\\y \end{pmatrix} \in\operatorname{span}\begin{pmatrix}
 X_u(\epsilon^{1/3},E) \\
  -{\sqrt{U(t_0(E),E)}}Y_u(\epsilon^{1/3},E)
\end{pmatrix}\right\},\\
W^s(\epsilon,E) \cap \{t=t_0(E)\} &=\left\{\begin{pmatrix}x\\y \end{pmatrix} \in \operatorname{span}\begin{pmatrix}
 X_s(\epsilon^{1/3},E) \\
  {\sqrt{U (t_0(E),E)}}Y_s(\epsilon^{1/3},E)
\end{pmatrix}\right\},
\end{align*}
where
\begin{align*}
 X_u(\epsilon^{1/3},E) = &\cos\left(\frac{1}{\epsilon }\int_{t_-(E)}^{t_0(E)} \sqrt{U(t,E)}dt -\frac{\pi}{4}+ \epsilon^{2/3} \phi_{1,u}(\epsilon^{1/3},E)\right),\\
 Y_u(\epsilon^{1/3},E) =& 
 (1+\epsilon^{2/3}\rho_u(\epsilon^{1/3},E))\times \\
 &\sin\left(\frac{1}{\epsilon }\int_{t_-(E)}^{t_0(E)} \sqrt{U(t,E)}dt -\frac{\pi}{4}+ \epsilon^{2/3} \phi_{2,u}(\epsilon^{1/3},E)\right),\\
 X_s(\epsilon^{1/3},E) = &\cos\left(\frac{1}{\epsilon }\int_{t_0(E)}^{t_+(E)} \sqrt{U(t,E)}dt -\frac{\pi}{4}+ \epsilon^{2/3} \phi_{1,s}(\epsilon^{1/3},E)\right),\\
 Y_s(\epsilon^{1/3},E) =& 
 (1+\epsilon^{2/3}\rho_s(\epsilon^{1/3},E))\times \\
 &\sin\left(\frac{1}{\epsilon }\int_{t_0(E)}^{t_+(E)} \sqrt{U(t,E)}dt -\frac{\pi}{4}+ \epsilon^{2/3} \phi_{2,s}(\epsilon^{1/3},E)\right),
\end{align*}
 and where
$\rho_{i},\phi_{1,i},\phi_{2,i}:[0,\epsilon_0^{1/3})\times D\rightarrow \mathbb R$, $i=u,s$ are all \KUK{$C^\infty$}-smooth functions.
\end{theorem}
  \begin{proof}
  \KUK{The expansions of $W^u\cap\{t=t_0(E)\}$ and $W^s\cap \{t=t_0(E)\}$ follow from the previous arguments, with $\rho_i,\phi_{1,i},\phi_{2,i}$, $i=u,s$, each being $C^M$-smooth with respect to $\epsilon^{1/3} \in [0,\epsilon_0^{1/3}(M))$ for any $M\in \mathbb N$. It is left to show that the functions are in fact $C^\infty$. 
  For this we use that $W^u$ and $W^s$ are uniquely fixed as the unstable and stable manifolds at infinity $t=\pm \infty$ (recall the discussion following assumption \ref{growth}). Hence $\rho_i,\phi_{1,i},\phi_{2,i}$, $i=u,s$ are independent of $M\in \mathbb N$ and it follows that they are each $C^\infty$ with respect to $\epsilon^{1/3}$ at $\epsilon^{1/3}=0$. At the same time, the problem is regular for $\epsilon>0$ and consequently $W^u\cap\{t=t_0(E)\}$ and $W^s\cap \{t=t_0(E)\}$ are each $C^\infty$ with respect to $\epsilon>0$ (and therefore also $C^\infty$ with respect to $\epsilon^{1/3}>0$). This shows that $\rho_i,\phi_{1,i},\phi_{2,i}$, $i=u,s$ are $C^\infty$ with respect to $\epsilon^{1/3}\in [0,\epsilon_0^{1/3})$ as desired.  }
%
  \end{proof}

\section{Proof of \thmref{01eigenvalues}}\seclab{01eigenvalues}
To prove \thmref{01eigenvalues} we first use \thmref{DeltaThm} and the characterization of the unstable and stable manifolds within the elliptic regime. It turns out to be useful to define the section $t=t_0(E)$ in the following way:
\begin{align}
 \int_{t_-(E)}^{t_0(E)} \sqrt{U(t,E)}dt  = \int_{t_0(E)}^{t_+(E)} \sqrt{U(t,E)}dt.\eqlab{t0Econd}
\end{align}
\begin{lemma}
 $t_0:\mathbb R_+\rightarrow \mathbb R$ defined by \eqref{t0Econd} is uniquely defined, $C^\infty$-smooth and satisfies
 \begin{align*}
  t_-(E)<t_0(E)<t_+(E),
 \end{align*}
for all $E$.
\end{lemma}
\begin{proof}
 With $t_- (E)<0<t_+(E)$ given, consider 
 \begin{align*}
  F(t,E): = \int_{t_-(E)}^{t} \sqrt{U(s,E)}ds - \int_{t}^{t_+(E)} \sqrt{U(s,E)}ds.
 \end{align*}
Clearly, $F(t,E)$ is well-defined for each $t\in [t_-(E),t_+(E)]$ and roots of $F(\cdot,E)$ correspond to solutions of \eqref{t0Econd}. We have $F(t_-(E),E)<0<F(t_+(E),E)$ and $F'_t(t,E)=2\sqrt{U(t,E)}>0$ for all $t\in (t_-(E),t_+(E))$. Consequently, there exists a unique $t=t_0(E)\in (t_-(E),t_+(E))$ such that 
\begin{align*}
 F(t_0(E),E)=0,
\end{align*}
for all $E>0$. The proof of the smoothness of $t_0(E)$ is similar to the proof of \lemmaref{Jprop}, see \appref{t0E}, and therefore left out.
\end{proof}

\begin{lemma}\lemmalab{DeltaThm}
 Consider \eqref{xyt} and suppose that \ref{Vquadratic}, \ref{smoothness} and \ref{growth} all hold and fix \KUK{$D\subset \mathbb R_+$} a compact interval. Then there exists an $\epsilon_0>0$ such that the following holds for any $\epsilon\in (0,\epsilon_0)$: $W^u(\epsilon,E)=W^s(\epsilon,E)$ for $E\in D$ if and only if
 \begin{align}
  \Delta(\epsilon,E) = 0,\eqlab{DeltaEq0}
 \end{align}
where 
 \begin{align}\eqlab{DeltaExpr}
 \Delta(\epsilon,E) = \sin\left(\frac{1}{\epsilon} \int_{t_-(E)}^{t_+(E)} \sqrt{U(t,E)} -\frac{\pi}{2}+\epsilon^{2/3} \phi(\epsilon^{1/3},E)\right)+\epsilon^{2/3}\rho(\epsilon^{1/3},E).
 \end{align}
 Here $\phi,\rho:[0,\epsilon_0^{1/3})\times D\rightarrow \mathbb R$ are both $C^\infty$-smooth functions.
 \end{lemma}
\begin{proof}
By \thmref{DeltaThm}, we have that $W^u(\epsilon,E)=W^s(\epsilon,E)$ if and only if 
\begin{align*}
\operatorname{det}\begin{pmatrix}
                   X_u(\epsilon^{1/3},E) & X_s(\epsilon^{1/3},E)\\
                   -Y_u(\epsilon^{1/3},E) & Y_s(\epsilon^{1/3},E)
                  \end{pmatrix}=0.
\end{align*}
By expanding the left hand side, using the definition of $t_0(E)$ in \eqref{t0Econd} we obtain
\begin{equation}\eqlab{deltaexpr}
\begin{aligned}
 &\frac12 \sin\left(\frac{1}{\epsilon }\int_{t_-(E)}^{t_+(E)} \sqrt{U(t,E)}dt -\frac{\pi}{2}+\mathcal O_1(\epsilon^{2/3})\right)\left(1 +\mathcal O_2(\epsilon^{2/3})\right)+\\
 &\frac12 \sin\left(\frac{1}{\epsilon }\int_{t_-(E)}^{t_+(E)} \sqrt{U(t,E)}dt -\frac{\pi}{2}+\mathcal O_3(\epsilon^{2/3})\right)\left(1 +\mathcal O_4(\epsilon^{2/3})\right)+\mathcal O_5(\epsilon^{2/3}),
\end{aligned}
\end{equation}
using subscripts to indicate that the $\mathcal O(\epsilon^{2/3})$-terms (each being $C^\infty$-smooth with respect to $\epsilon^{1/3}$ and $E$) are different (in general). (The advantage of defining $t_0(E)$ as in \eqref{t0Econd} is that terms with 
\begin{align*}
 \frac{1}{\epsilon }\int_{t_-(E)}^{t_0(E)} \sqrt{U(t,E)}dt-\frac{1}{\epsilon }\int_{t_0(E)}^{t_+(E)} \sqrt{U(t,E)}dt,
\end{align*}
disappear.)
Using simple trigonometric identities, we then write the first two terms of \eqref{deltaexpr} as $$\frac12 (1+\epsilon^{2/3}\rho(\epsilon^{1/3},E) )\sin\left(\frac{1}{\epsilon }\int_{t_-(E)}^{t_+(E)} \sqrt{U(t,E)}dt -\frac{\pi}{2}+\epsilon^{2/3}\phi(\epsilon^{1/3},E)\right),$$
with $\rho$ and $\phi$ defined by
\begin{align*}
 \left(1+\epsilon^{2/3}\rho(\epsilon^{1/3},E)\right)^2 : &=\left((1 +\mathcal O_2) \cos \mathcal O_1 +(1 +\mathcal O_4) \cos \mathcal O_3\right)^2\\
 &+\left((1 +\mathcal O_2) \sin \mathcal O_1 +(1 +\mathcal O_4) \sin \mathcal O_3 \right)^2,\\
 \epsilon^{2/3} \phi(\epsilon^{1/3},E):&=\tan^{-1}\left(\frac{(1 +\mathcal O_2) \sin \mathcal O_1 +(1 +\mathcal O_4) \sin \mathcal O_3 }{(1 +\mathcal O_2) \cos \mathcal O_1 +(1 +\mathcal O_4) \cos \mathcal O_3}\right).
\end{align*}
 Subsequently, we divide the resulting expression for \eqref{deltaexpr} by $\frac12 (1+\epsilon^{2/3}\rho(\epsilon^{1/3},E)) $ for all $\epsilon>0$ sufficiently small.  In this way, we obtain \eqref{DeltaEq0} and \eqref{DeltaExpr}.


\end{proof}

We now use the fact that $J$ is a $C^\infty$ diffeomorphism, recall \lemmaref{Jprop}, to write the equation in terms of $\epsilon$ and $J$ instead:
\begin{align}\eqlab{DeltaExpr2}
\widetilde \Delta(\epsilon,J)= \sin \left(\frac{\pi}{\epsilon}J -\frac{\pi}{2}+\epsilon^{2/3}\widetilde \phi(\epsilon^{1/3},J) \right) +\epsilon^{2/3}\widetilde\rho(\epsilon^{1/3},J),
\end{align}
for $\epsilon \in (0,\epsilon_0)$, $J\in \widetilde D=E(D)$,
where $\widetilde \phi(\epsilon^{1/3},J)=\phi(\epsilon^{1/3},E(J))$, $\widetilde \rho(\epsilon^{1/3},J)=\rho(\epsilon^{1/3},E(J))$ and where $E(J)$ denotes the $C^\infty$-smooth inverse of \eqref{Jdefn}. 

We henceforth drop the tilde in \eqref{DeltaExpr2}. In the following, we proceed to solve \eqref{DeltaExpr2} for $J$ as a function of $\epsilon$. 
\begin{lemma}\lemmalab{solveDeltaEq0}
 $\Delta(\epsilon,J)=0$, $\epsilon \in (0,\epsilon_0)$, $J\in D$, if and only if 
 \begin{align}
  J =\delta+\frac12 \epsilon + \rsp{\epsilon^{5/3}} \overline\phi(\epsilon^{1/3},J,(-1)^n),\eqlab{Jfinal}
 \end{align}
 where 
 \begin{align}
  \delta = n \epsilon,\quad n\in \mathbb N.
 \end{align}
 Here each of the functions:
 \begin{align*}
  \overline \phi(\cdot,\cdot,\pm 1):[0,\epsilon_0^{1/3})\times D\rightarrow \mathbb R,
 \end{align*}
are $C^\infty$-smooth.
 
%
\end{lemma}
\begin{proof}
Let $K$ be so that
\begin{align}
 J =\delta+\frac{1}{2} \epsilon - \frac{1}{\pi}\rsp{\epsilon^{5/3}} \phi(\epsilon^{1/3},J)+\frac{\epsilon}{\pi} K.\eqlab{Jsubs}
\end{align}
Inserting this into $\Delta(\epsilon,J)=0$, see \eqref{DeltaExpr2}, gives
 \begin{align}
  \sin(n \pi + K)+\epsilon^{2/3}\rho(\epsilon^{1/3},J)=(-1)^n\sin(K)+\epsilon^{2/3}\rho(\epsilon^{1/3},J)=0,\eqlab{Keqn}
 \end{align}
 upon using that $\delta=n\epsilon$. We obtain a solution 
 \begin{align*}
  K=\epsilon^{2/3} \widetilde K(\epsilon^{1/3},J,(-1)^n),
 \end{align*}
 of \eqref{Keqn},  for some $C^\infty$-smooth $\widetilde K$,  for all $\epsilon\ge 0$ small enough. Inserting this into \eqref{Jsubs} gives \eqref{Jfinal}.
%
\end{proof}
We now solve \eqref{Jfinal} for $J\in D$ as a function of $\delta$ and $\epsilon$. 
\begin{lemma}
Fix any \KUK{$D_0\subsetneq D$} a compact interval. Then there exists an $\epsilon_0>0$ sufficiently small so that the following holds for any $\delta\in D_0$, $\epsilon\in [0,\epsilon_0^{1/3})$: There exists a unique solution of \eqref{Jfinal} of the following form:
\begin{align*}
 J &= \delta +\frac12 \epsilon+\rsp{\epsilon^{5/3}} \rsp{\overline J}(\delta,\epsilon^{1/3},(-1)^n),
\end{align*}
where
\begin{align}
\rsp{\overline J}(\cdot,\cdot,\pm 1): D_0\times [0,\epsilon_0^{1/3})\rightarrow \mathbb R,
\end{align}
are $C^\infty$-smooth functions. 
\end{lemma}
\begin{proof}
 The result clearly follows from the implicit function theorem: $J=\delta$ is a regular solution of \eqref{Jfinal}  for $\epsilon=0$. 
\end{proof}
In this way, we have proven \thmref{01eigenvalues}.

\section{Proof of \thmref{smalleigenvalues}}\seclab{smalleigenvalues}

\rsp{To prove \thmref{smalleigenvalues}, we insert}
\begin{align}
E=\epsilon e,\eqlab{scalinge}
\end{align}
into \eqref{xyt}. This produces the extended system
 \begin{equation}\eqlab{systemsmall}
\begin{aligned}
 \dot x &= y,\\
 \dot y &=(V(t)-\epsilon e)x,\\
 \dot t &=\epsilon,\\
 \dot \epsilon &=0.
\end{aligned}
\end{equation}
We consider $e\in D_1$ with $D_1$ fixed in some large compact set and $0\le \epsilon\ll 1$.
By assumption \ref{Vquadratic} each point $(0,0,t,0)$ with $t\ne 0$ and $\epsilon=0$ is partially hyperbolic for all $e\in D_1$, having stable and unstable manifolds $W^s$ and $W^u$. However, $(x,0,0,0)$ is fully nonhyperbolic for all $x\in \mathbb R$. We therefore apply the following blowup transformation:
 \begin{align}
  (r,(\bar y,\bar t,\bar \epsilon))\mapsto \begin{cases} 
                             y = r\bar y,\\
                             t\,=r \bar t,\\
                             \epsilon \,=r^2 \bar \epsilon,
                            \end{cases}\eqlab{blowupsmall}
 \end{align}
where $r\ge 0$, $(\bar y, \bar t,\bar \epsilon)\in S^2$. Here $S^2\subset \mathbb R^3$ denotes the unit sphere. In this way, the set of points $(x,0,0,0)$ is blown up to a cylinder of spheres. To describe the blown up system we use two charts: $\bar t=-1$ and $\bar \epsilon=1$ with chart-specific coordinates $(y_1,r_1,\epsilon_1)$ and $(y_2,t_2,r_2)$, respectively, defined by 
 \begin{align}\eqlab{bartEq1}
(y_1,r_1,\epsilon_1)\mapsto 
 \begin{cases} y = r_1 y_1,\\
 t \,= -r_1,\\
\epsilon \,= r_1^2\epsilon_1,
\end{cases}
\end{align}
and
\begin{align}\eqlab{barEpsEq1}
(r_2,y_2,t_2)\mapsto 
 \begin{cases} y = r_2 y_2,\\
 t \,= r_2 t_2,\\
\epsilon \,= r_2^2.
\end{cases}
\end{align}
The charts overlap for $t_2<0$ and the change of coordinates is given by
\begin{align}
 &\begin{cases} r_1 = r_2 (-t_2),\\  y_1 = y_2 (-t_2)^{-1},\\\epsilon_1 \,= (-t_2)^{-2}.\end{cases}\eqlab{cc12A}
\end{align}
The details of the corresponding $\bar t=1$-chart are similar to those in $\bar t=-1$ and the details of $\bar t=1$ are therefore left out.

We present further details below, but in summary the blowup transformation (by working in the charts $\bar t=-1$ and $\bar t=1$) allows us to extend the unstable manifolds into the scaling chart $\bar \epsilon=1$ where $t=\mathcal O(\epsilon^{1/2})$. 
In the scaling chart, we obtain the following Weber equation for $r_2=0$ (upon desingularization)
\begin{equation}\eqlab{webersystem}
\begin{aligned}
 \dot x &=y_2,\\
 \dot y_2 &= (t_2^2-e) x,\\
 \dot t_2 &=1.
\end{aligned}
\end{equation}
For this system the eigenvalues are known: $e=2n+1$ for all $n\in \mathbb N_0$, see \cite{AbramowitzStegun1964} and \eqref{weber2nd} below, and we obtain true eigenvalues $e_{0}(\epsilon),\ldots ,e_{n_0}(\epsilon)$, where $e_{k} = 2k+1$, $k\in \{0,1,\ldots,n_0\}$, of \eqref{xyt} with $E=\epsilon e$, for $e\in D_1$, for all $0<\epsilon\ll 1$ by applying (regular) perturbation theory. 



\subsection{Analysis in the $\bar t=-1$-chart}
Upon inserting \eqref{bartEq1} into the extended system \eqref{systemsmall}, we obtain the following equations
\begin{align*}
 \dot x &=y_1,\\
 \dot y_1 &=(V_1(r_1)-\epsilon_1 e) x+\epsilon_1 y_1,\\
 \dot r_1 &=-r_1\epsilon_1,\\
 \dot \epsilon_1 &=2\epsilon_1^2,
\end{align*}
upon desingularization through division by $r_1$. Here
\begin{align*}
 V_1(r_1): = r_1^{-2} V(-r_1),
\end{align*}
by
 assumption \ref{Vquadratic},
has a $C^\infty$-smooth extension to $r_1=0$ given by $V_1(0)=1$. 

Now, setting $r_1=\epsilon_1=0$ gives 
\begin{align*}
 \dot x &=y_1,\\
 \dot y_1 &= x.
\end{align*}
For this system $(x,y_1)=(0,0)$ is a linear saddle with stable space $\text{span}\,(1,-1)$ and unstable space $\text{span}\,(1,1)$. For the full $(x,y_1,r_1,\epsilon_1)$-space, the set defined by $(0,0,r_1,\epsilon_1)$ with $r_1$, $\epsilon_1\sim 0$, becomes a center manifold, having a smooth foliation of stable and unstable fibers. 
Each foliation produces a stable manifold $W_1^s$ and an unstable manifold $W_1^u$ of the following local graph form
\begin{align*}
W_1^s:\quad &y_1 = \left(-1 +h_s(\epsilon_1,r_1,e)\right)x,\\
 W_1^u:\quad &y_1 = \left(1 +h_u(\epsilon_1,r_1,E)\right)x,
\end{align*}
for some smooth $h_{s,u}$ with $h_{s,u}(0,0,e)=0$ for all $e\in D_1$, and $(\epsilon_1,r_1)\in [0,\upsilon]^2$. (By center manifold theory, the smoothness of $h_{s,u}$ is apriori only finite for fixed $\upsilon>0$ but arbitrary as $\upsilon\rightarrow 0$.) 
\begin{lemma}\lemmalab{W1uW1s}
 $W_1^u\cap \{r_1=0\}$ is unique whereas $W_1^s\cap \{r_1=0\}$ is nonunique.
\end{lemma}
\begin{proof}
 Let $u:=x^{-1} y_1$. Then
 \begin{align*}
  \dot u &= 1-u^2 -\epsilon_1 e+\epsilon_1 u,\\
  \dot \epsilon_1 &= 2 \epsilon_1^2.
 \end{align*}
 $W_1^u\cap \{r_1=0\}$ and $W_1^s\cap \{r_1=0\}$ are within this projective space center manifolds of $(u,\epsilon_1)=(1,0)$ and $(u,\epsilon_1)=(-1,0)$, respectively. The center manifold of the former is unique since it is a nonhyperbolic saddle, whereas the latter is a nonhyperbolic unstable node. 
\end{proof}

\subsection{Analysis in the $\bar \epsilon=1$-chart}
Upon inserting \eqref{barEpsEq1} into \eqref{systemsmall}, we obtain the following equations:
   \begin{equation}\eqlab{xy2t2}
 \begin{aligned}
 \dot x &=y_2,\\
 \dot y_2 &= (t_2^2+r_2 t_2^3 R_2(r_2 t_2)-e) x,\\
 \dot t_2 &=1,
\end{aligned}
\end{equation}
and $\dot r_2=0$ upon desingularization through division by $r_2$. Here $R_2$ is a $C^\infty$-smooth function.

By using \eqref{cc12A} we can transform the result \rspp{of} \lemmaref{W1uW1s} on the stable and unstable manifolds into the present chart. This gives an unstable manifold $W_2^u(r_2,e)$ in the $(x,y_2,t_2)$-space, which within compact subsets, depends smoothly on $r_2$ and $e$. In particular, within compact subsets it is smoothly $\mathcal O(r_2)$-close, uniformly in $e\in D_1$, to the unique unstable manifold of the $r_2=0$ subsystem, see \eqref{webersystem}, for $t_2\ll -1$. By working in the $\bar t=1$-chart, then we obtain a similar result on the stable manifold $W_2^s(r_2,e)$. Here $W_2^s(r_2,e)$ is also (within compact subsets of the $(x,y_2,t_2)$-space) smoothly $\mathcal O(r_2)$-close, uniformly in $e\in D_1$, to a unique stable manifold of the $r_2=0$ subsystem for $t_2\gg 1$.

We can write \eqref{xy2t2} for $r_2=0$ as a first order system:
\begin{align}
x''(t_2) =(t_2^2-e) x(t_2),\eqlab{weber2nd}
\end{align}
or upon setting $x(t_2)= e^{-\frac12 t_2^2} u(t_2)$:
\begin{align*}
u''(t_2) -2tu'(t_2)+(e-1) u=0.
\end{align*}
This equation has a polynomial solution $u(t_2)=H_n(t_2)$ for $e=2n+1$,  $H_n$ is the Hermite polynomial of degree $n\in \mathbb N_0$, see e.g. \cite{AbramowitzStegun1964}.
\begin{lemma}
\begin{align}
W_2^u(0,e) = W_2^s(0,e) \Longleftrightarrow e=2n+1,\,n\in \mathbb N_0.
\end{align}
Specifically, if $e=2n+1$, $n\in \mathbb N_0$ then 
\begin{equation}\eqlab{webersol}
\begin{aligned}
x(t_2) &= e^{-\frac12 t_2^2} H_n(t_2),\\
y_2(t_2) &= e^{-\frac12 t_2^2}\left(2nH_{n-1}(t_2)-t_2 H_n(t_2)\right),
\end{aligned}
\end{equation}
is a bounded solution of \eqref{webersystem}. 
\end{lemma}
We now use that $W_2^u(r_2,e)$ and $W_2^s(r_2,e)$ are smooth perturbations of $W_2^{u}(0,e)$ and $W_2^s(0,e)$ (in any compact domain of the $(x_2,y_2,t_2)$-space) to solve for true eigenvalues. There are different ways to proceed but our approach is inspired by Melnikov theory, see e.g. \cite{Guckenheimer97}.

Fix $n\in \mathbb N_0$, write $$e=2n+1+\tilde e,$$ and let $\Sigma_2$ denote the section at $t_2=0$. Moreover, let $\mathcal U_2$ denote the intersection of $W_2^u(0,2n+1)=W_2^s(0,2n+1)$ with $\Sigma_2$. By \eqref{webersol}, and the fact that $H_n$ is even/odd if $n$ is so, $\mathcal U_2$ is given by 
the $x$-axis if $n$ is even and the $y_2$-axis if $n$ is odd. Let $(x^i(t_2,r_2,\tilde e),y_2^i(t_2,r_2,\tilde e),t_2)\in W_2^i(r_2,2n+1+\tilde e)$ for $i=u,s$ denote solutions of \eqref{xy2t2} with $x^i(0,r_2,\tilde e)=1$ if $n$ is even and $y_2^i(0,r_2,\tilde e)=1$ if $n$ is odd. We put $\eta:= (x^i(0,r_2,\tilde e),y_2^i(0,r_2,\tilde e))$, which only depends upon $n$. Finally, we let $\mathcal V_2$ denote the one-dimensional space within $\Sigma_2$ that is orthogonal to $\mathcal U_2$ and define $\mathcal P$ as the orthogonal projection onto the $\mathcal V_2$-space.

Finally, define
\begin{align*}
 q(t_2,r_2,\tilde e) := \begin{cases}
                        q^u(t_2,r_2,\tilde e) &\mbox{for}\,t_2\le 0,\\
                        q^s(t_2,r_2,\tilde e) &\mbox{for}\,t_2\ge 0,
                       \end{cases}
\end{align*}
for $q=x,y_2$. $(x(t_2,r_2,\tilde e),y_2(t_2,r_2,\tilde e),t_2)$ is clearly only a solution of \eqref{xy2t2} for all $t_2\in \mathbb R$ if $W_2^u=W_2^s$. 
\begin{lemma}\lemmalab{Delta1}
For all $r_2$ and $\tilde e$ sufficiently small, $W^u(r_2,2n+1+\tilde e)=W_2^s(r_2,2n+1+\tilde e)$ if and only if $\Delta_1(r_2,\tilde e)=0$ where
\begin{align*}
\Delta_1(r_2,\tilde e):=\int_{-\infty}^\infty e^{-\frac12 t_2^2} H_n(t_2) \left[-\tilde e +r_2 t_2^3 R_2(r_2t_2) \right] x(t_2,r_2,\tilde e) dt_2.
\end{align*}
\end{lemma}
\begin{proof}
 We write $$z:=(x,y_2),$$ and \eqref{xy2t2} as follows
 \begin{align*}
 \dot z &= A(t_2)z+r_2 \begin{pmatrix}
                        0\\
                        t_2^3 R_2(t_2r_2) x,
                       \end{pmatrix}
%
 \end{align*}
 and
 let $\Phi(t_2,t_{20})$ denote the state-transition matrix of the $r_2=0$-subsystem: $\dot z=A(t_2)z$. In this way, it is then standard to write $$z^u(0,r_2,\tilde e):=(x^u(0,r_2,\tilde e),y_2^u(0,r_2,\tilde e)),$$ in the following form:
 \begin{align*}
  z^u(0,r_2,\tilde e) = \eta + \mathcal P\int_{-\infty}^0 \Phi(0,t_2) \begin{pmatrix}
                        0\\
                        t_2^3 R_2(t_2r_2) x^u(t_2,r_2,\tilde e)                
                       \end{pmatrix} dt_2.
 \end{align*}
We obtain a similar expression for $z^s(0,r_2,\tilde e):=(x^s(0,r_2,\tilde e),y_2^s(0,r_2,\tilde e))$:
\begin{align*}
 z^s(0,r_2,\tilde e) = \eta+ \mathcal P\int_{\infty}^0 \Phi(0,t_2) \begin{pmatrix}
                        0\\
                        t_2^3 R_2(t_2r_2) x^s(t_2,r_2,\tilde e)                
                       \end{pmatrix} dt_2.
\end{align*}
Finally, we use that $$\psi(t_2) =\begin{pmatrix}
                                   -y_2(t_2,0,0)\\
                                   x_2(t_2,0,0)
                                  \end{pmatrix},$$
                                  is the unique (up to scalar multiplication) bounded solution of the adjoint equation $\dot w = -A(t_2)^T w$, satisfying $\psi(0)=\eta^\perp$. Therefore upon taking the dot product of $\psi(0)$ with $(z^u-z^s)(0,r_2,\tilde e)$, and using that $\psi(0)^T \mathcal P\Phi(0,t_2) = \psi(t_2)^T$, we obtain the desired result. 
\end{proof}
We now solve $\Delta_1(r_2,\tilde e)=0$. We have $\Delta_1(0,0)=0$ and
\begin{align*}
 \frac{\partial }{\partial \tilde e}  \Delta_1(0,0) = -\int_{-\infty}^\infty e^{-t_2^2} H_n(t_2)^2 dt_2 = -\sqrt{\pi} 2^n n!\ne 0,
\end{align*}
see e.g. \cite{AbramowitzStegun1964}.
Consequently, by the implicit function theorem, there exists a $C^\infty$-smooth function $\hat e:[0,r_{20})\rightarrow \mathbb R$ such that $\hat e(0)=0$ and $\Delta_1(r_2,\hat e(r_2))=0$ for all $r_2\in [0,r_{20})$. The following lemma, shows that $\hat e$ is actually a $C^\infty$-smooth function of $\epsilon=r_2^2$.
\begin{lemma}
\eqref{xy2t2} is invariant with respect to $(r_2,t_2)\mapsto (-r_2,t_2)$ upon time reversal.
\end{lemma}
In this way, eigenvalues $\hat e(r_2)$ for $r_2<0$ sufficiently small satisfy $\hat e(r_2)=\hat e(-r_2)$. Therefore $\hat e$ is even as a function of $r_2$ and consequently a $C^\infty$-smooth function of $\epsilon =r_2^2$. This completes the proof of \thmref{smalleigenvalues}.

\section{Proof of \thmref{intermediateEigenvalues}}\seclab{intermediateEigenvalues}

\rsp{Consider the blowup in parameter space of $(\epsilon,E)=(0,0)$ defined by
\begin{align}
 (\rho,(\bar \epsilon,\overline E))\mapsto \begin{cases}
                                       \epsilon \,\,= \rho \bar \epsilon,\\
                                       E =\rho \overline E,
                                      \end{cases}\eqlab{bu_par}
\end{align}
for $\rho\ge 0$, $(\bar \epsilon,\overline E)\in S^1\subset \mathbb R^2$. The scaling \eqref{scalinge}, used to prove \thmref{smalleigenvalues}, can be interpreted as the $\bar \epsilon=1$-chart associated with \eqref{bu_par}:
\begin{align*}
 (\rho_1,e)\mapsto \begin{cases}
                      \epsilon \,\,=\rho_1,\\
                      E =\rho_1e,
                     \end{cases}
\end{align*}
upon eliminating $\rho_1\ge 0$. In this section, we will work in 
the $\overline E=1$-chart associated with \eqref{bu_par}:
\begin{align*}
 (\rho_2,\xi )\mapsto \begin{cases}
                      \epsilon \,\,=\rho_2 \xi,\\
                      E =\rho_2,
                     \end{cases}
\end{align*}
or simply
\begin{align}
 \epsilon &= E \xi,\eqlab{epsxi}
\end{align}
upon eliminating $\rho_2 \ge 0$. Notice that 
\begin{align}
\xi= e^{-1}.\eqlab{xie}
\end{align}
Consequently, eigenvalues $e$ obtained in the $\bar \epsilon = 1$-chart, are also visible in the  $\overline E=1$-chart. At the same time, a subset of the eigenvalues of \thmref{01eigenvalues} are clearly also visible there (taking $D=[c_0,c_1]$ with $c_0>0$ small enough in \thmref{01eigenvalues}). Therefore the $\overline E=1$-chart is ideally suited for describing the intermediate eigenvalues.}

\rsp{Inserting \eqref{epsxi} into \eqref{xyt}
 gives the extended system}
 \begin{equation}\eqlab{xytrho}
 \begin{aligned}
  \dot x &= y,\\
  \dot y&=(V(t)-E)x,\\
  \dot t &= E\xi,\\
  \dot E&=0,
 \end{aligned}
 \end{equation}
obtained from \eqref{xyt}. \rsp{We will study \eqref{xytrho} in the following. It suffices to consider $(\xi,E)\in D_2$ with $D_2$ a small neighborhood of $(0,0)$.} 
We have $V(0)=V'(0)=0$, $V''(0)=2$ by assumption \ref{Vquadratic}. Consequently, each point $(x,0,0,0)$ is degenerate for \eqref{xytrho}, the linearization having only zero eigenvalues. We therefore proceed as in \secref{smalleigenvalues} and apply the blowup transformation (similar to \eqref{blowupsmall}):
 \begin{align*}
  (r,(\bar y,\bar t,\overline E))\mapsto \begin{cases} 
                             y \,= r\bar y,\\
                             t\,\,=r \bar t,\\
                             E =r^2 \overline E
                            \end{cases}
 \end{align*}
where $r\ge 0$, $(\bar y, \bar t,\overline E)\in S^2\subset \mathbb R^3$. In this way, the set of points $(x,0,0,0)$ is blown up to a cylinder of spheres. Proceeding as in \secref{smalleigenvalues}, we can extend the hyperbolicity and control the stable and unstable manifolds (by working in the directional charts $\bar t=\pm 1$) all the way into the corresponding scaling chart, obtained by setting $\overline E=1$:
\begin{equation}\eqlab{chart2intermediate}
\begin{aligned}
(r_2,y_2,t_2)\mapsto 
 \begin{cases} y \,= r_2 y_2,\\
 t \,\,= r_2 t_2,\\
E = r_2^2.
\end{cases}
\end{aligned}
\end{equation}
(Here we abuse notation slightly by using the same symbols as in \eqref{barEpsEq1} for different coordinates). 
In this chart, we obtain the following equations
\begin{equation}\eqlab{xy2t2int}
\begin{aligned}
 \dot x &=y_2,\\
 \dot y_2 &= -U_2(t_2,r_2)  x,\\
 \dot t_2 &=\xi,
\end{aligned}
\end{equation}
and $\dot r_2=0$,
upon dividing the right hand side by $r_2$ (desingularization). Here 
\begin{align*}
U_2(t_2,r_2):=1-r_2^{-2} V(r_2t_2),
\end{align*}
which has a $C^\infty$-smooth extension to $r_2=0$ given by $U_2(t_2,0)=1-t_2^2$ by assumption \ref{Vquadratic}.

Let $W^u(\xi,r_2)$ and $W^s(\xi,r_2)$ denote the unstable and stable manifolds of $x=y_2=0$ for \eqref{xy2t2int} for $t_2\ll -1$ and $t_2\gg 1$, respectively, and $0\le r_2<r_{20}$, $0<\xi<\xi_0$, $r_{20},\xi_0$ both sufficiently small. We denote their extensions to all $t_2\in \mathbb R$ by the same symbol.

Let $t_{2\mp }(r_2)$ be so that $U(t_\mp (r_2),r_2) = 0$ and $t_{2\mp}(0)=\mp 1$. $t_{2\pm}$ are $C^\infty$-smooth functions of $r_2\ge 0$ small enough by the implicit function theorem.
\begin{lemma}
 \KUK{The following holds for all $\xi\in (0,\xi_0)$ and $r_2\in [0,r_{20})$}:
 \begin{align*}
\KUK{W^u(\xi,r_2)=W^s(\xi,r_2)\quad \Longleftrightarrow \quad \Delta_2(\xi,r_2)=0},
 \end{align*} where
\begin{equation}\eqlab{Delta2Expr}
\begin{aligned}
\Delta_2(\xi,r_2) &=  \sin \left(\frac{1}{\xi}\int_{t_{2-}(r_2)}^{t_{2+}(r_2)} \sqrt{U_2(t_2,r_2)} dt_2 -\frac{\pi}{2}+\xi^{2/3}\phi_2(\xi^{1/3},r_2) \right) \\
&+\xi^{2/3}\rho_2(\xi^{1/3},r_2),
 \end{aligned}
 \end{equation}
for some $C^\infty$-smooth functions $\phi_2,\rho_2:[0,\xi_0^{1/3})\times [0,r_{20})\rightarrow \mathbb R$. Moreover, 
\begin{align}
 \Delta_2(\xi,0)=0\quad \Longleftrightarrow\quad \exists\, \KUK{n\in \mathbb N_0\,:\, \xi=(2n+1)^{-1}}.
 \eqlab{Delta2n}
\end{align}
\end{lemma}
\begin{proof}
  \lemmaref{DeltaThm} applies to \eqref{xy2t2int} with $\epsilon=\xi>0$ small enough and $U=U_2$. Therefore \eqref{DeltaExpr} with $\epsilon=\xi$ gives \eqref{Delta2Expr}. 
  
  Next regarding \eqref{Delta2n}, we realize that \eqref{xy2t2int} for $r_2=0$ is equivalent to the Weber equation \eqref{weber2nd} upon setting $\xi = e^{-1}$, recall \eqref{xie}. Since $e=2n+1$, \KUK{$n\in \mathbb N_0$}, are the eigenvalues associated with the Weber equation, we obtain \eqref{Delta2n}.

\end{proof}
Following this lemma, we proceed to solve $\Delta_2(\xi,r_2)=0$, using the same approach as for the large eigenvalues: 
\begin{lemma}
 $\Delta_2(\xi,r_2) =0$, $\xi\in (0,\xi_0), r_2\in [0,r_{20}]$ if and only if there is some $n\in \mathbb N$ such that 
 \begin{align}
  \frac{1}{\pi}\int_{t_{2-}(r_2)}^{t_{2+}(r_2)} \sqrt{U_2(t_2,r_2)} dt_2&=  \left(n+\frac{1}{2}\right)\xi+\xi^{5/3}\rsp{\overline J}_2(\xi^{1/3},r_2,(-1)^n),\eqlab{solutionDelta20}
 \end{align}
 where each of the functions
 \begin{align*}
  \rsp{\overline J}_2(\cdot,\cdot,\pm 1):[0,\xi_0)\times [0,r_{20}]\rightarrow \mathbb R,
 \end{align*}
are \KUK{$C^\infty$}-smooth, and satisfy:
\begin{align}
 \rsp{\overline J}_2((2n+1)^{-1/3},0,(-1)^n)=0,\eqlab{J2cond}
\end{align}
for all \KUK{$n\in \mathbb N_0$}. 
\end{lemma}
\begin{proof}
 We proceed as in the proof of \lemmaref{solveDeltaEq0} and define $K_2$ by
 \begin{align}
  \frac{1}{\xi}\int_{t_{2-}(r_2)}^{t_{2+}(r_2)} \sqrt{U_2(t_2,r_2)} dt_2 -\frac{\pi}{2}+\xi^{2/3}\phi_2(\xi^{1/3},r_2)=\pi n + K_2.\eqlab{K2defn}
 \end{align}
Inserting this into \eqref{Delta2Expr}, see \eqref{Delta2Expr}, gives
\begin{align*}
 (-1)^n \sin(K_2) +\xi^{2/3}\rho_2(\xi^{1/3},r_2)=0,
\end{align*}
which we solve for $K_2=\xi^{2/3} \widetilde K_2(\xi^{1/3},r_2,(-1)^n)$ for all $\xi\ge 0$ small enough. Inserting the resulting expression for $K_2$ into \eqref{K2defn} and rearranging gives 
\begin{align}
 \frac{1}{\pi}\int_{t_{2-}(r_2)}^{t_{2+}(r_2)} \sqrt{U_2(t_2,r_2)} dt_2&= \left( n +\frac{1}{2}\right)\xi+\xi^{5/3} \rsp{\overline J}_2(\xi^{1/3},r_2,(-1)^n),\eqlab{intermediateEigenvalues}
 \end{align}
 with each $\rsp{\overline J}_2(\cdot,\cdot,\pm 1)$ is $C^\infty$-smooth. Now, for $r_2=0$ we have $\int_{t_{2-}(r_2)}^{t_{2+}(r_2)} \sqrt{U_2(t_2,r_2)} dt_2=\frac{\pi}{2}$ and therefore
 \begin{align*}
 \frac{1}{2} &=  \left(n+\frac{1}{2}\right)\xi+\xi^{5/3} \rsp{\overline J}_2(\xi^{1/3},0,(-1)^n).
 \end{align*}
 By \eqref{Delta2n}, this gives $\xi=(2n+1)^{-1}$ for all $n\in \mathbb N_0$, $n\gg 1$, large enough. This proves \eqref{J2cond}.
 \end{proof}
 


Multiplying \eqref{intermediateEigenvalues} by $r_2^2$ gives 
\begin{align*}
J(r_2^2) = n\epsilon + \frac{1}{2}\epsilon + \xi^{5/3} r_2^2 \rsp{\overline J}_2(\xi^{1/3},r_2,(-1)^n).
\end{align*}
recall \eqref{Jdefn}. Setting $E=r_2^2$ and $\xi=E^{-1}\epsilon$, recall \eqref{chart2intermediate} and \eqref{epsxi}, gives the desired result, \thmref{intermediateEigenvalues}.

 


\section{Discussion} \seclab{discussion}
%
In this paper, we \rsp{revisited} the eigenvalue problem of the one-dimensional Schr{\"o}dinger equation, recall \eqref{eq0}, for $C^\infty$ potentials. In particular, we provide a new interpretation of the Bohr-Sommerfeld quantization formula.  A novel aspect of our results, which are based on the recent work \cite{uldall2024a}, is that
we cover all eigenvalues $E\in [0,E_0]$ for all $0<\epsilon\ll 1$; here $E_0>0$ is any fixed constant. For this purpose, we connect the small eigenvalues $E=\mathcal O(\epsilon)$ with large eigenvalues $E=\mathcal O(1)$ through intermediate eigenvalues and show that the Bohr-Sommerfeld quantitization formula approximates all of these eigenvalues (in a sense that is made precise). 

Our results also provide rigorous smoothness statements of the eigenvalues. Whereas the small eigenvalues $E=\mathcal O(\epsilon)$ are $C^\infty$-smooth functions of $\epsilon$, the large eigenvalues $E=\mathcal O(1)$ are $C^\infty$-smooth functions of $n\epsilon\in [c_1,c_2]$ and $\epsilon^{1/3}$; here $n\in \mathbb N$ is the index of the eigenvalues. The change in smoothness is common in problems with turning points, see e.g. \cite{freddy,gils2005a,krupa_extending_2001}. In fact, for the fold point \cite{krupa_extending_2001} the expansions of the slow manifold changes from being smooth with respect to $\epsilon$ to only being smooth with respect to $\epsilon^{\frac23}$ and $\epsilon^{1/3}\log \epsilon^{-1}$. \rspp{In} \cite{uldall2024a}, which is based on normal form theory, one would also expect presence of logarithms in the present case. But the detailed analysis in \cite{uldall2024a} shows that the corresponding ``resonant terms'' are absent. It is also possible that our smoothness results are suboptimal, but we have not found a way to improve this.

\KUK{Although, we assumed that $V\in C^\infty$, it is clear that $V\in C^k$ with $k\gg 1$ large enough will suffice. We leave the details of this case to the interested reader. Moreover, while our focus has been on a single well potential, it is not difficult to extend our results to more general potentials with additional minima, provided that the growth condition, see assumption \ref{growth}, holds at $t\rightarrow \pm \infty$. \KUK{This could include tunneling \cite{hislop1996a,simon1984a} (e.g. through the Exchange Lemma \cite{schecter2008a})}. However, in such situations, the description of eigenvalues near local maxima require a separate detailed description of a reversed type of turning point, corresponding to \eqref{airy1} with $\mu(0)=\mu'(0)=0$ and $\mu''(0)<0$. We will analyze this situation in a separate forthcoming paper \cite{kriszm22c}.}



\newpage 
\bibliography{refs}
\bibliographystyle{plain}
\newpage
\appendix 

\section{Proof of \lemmaref{Jprop}}\applab{t0E}

Let 
\begin{align*}
 U(t,E):=E-V(t).
\end{align*}
$t_\pm(E)$ are therefore roots of $U(\cdot,E)$. 
We have that 
\begin{align}
 \frac{\partial}{\partial t}U(t,E) = -V'(t).\eqlab{Ut}
\end{align}
Using the assumptions \ref{Vquadratic} and \ref{growth}, the sign of the right hand side of \eqref{Ut} is $-\operatorname{sign}(t)$ for all $t\in \mathbb R$, and the existence of the $C^\infty$-smooth functions $t_\pm(E)$ therefore follow. This completes the proof of item \ref{tpmE}.

Next, we turn to the proof of item \ref{Jprop2} and the smoothness of 
\begin{align*}
 J(E) = \int_{t_-(E)}^{t_+(E)} \sqrt{U(t,E)} dt. 
\end{align*}
By Leibniz's rule of differentiation we have that
\begin{align}
 J'(E) = \int_{t_-(E)}^{t_+(E)} \frac{1}{2\sqrt{U(t,E)}} dt>0,\eqlab{Jprime}
\end{align}
upon using the definitions of $U$ and $t_\pm (E)$.
Notice that the right hand side is integrable, since $t_\pm(E)$ are simple roots of $U(\cdot,E)$ and the limits
\begin{align*}
\lim_{t\rightarrow t_\pm(E)} \vert t-t_\pm (E)\vert^{-1/2} \sqrt{U(t,E)},
\end{align*}
are therefore well-defined. In turn, $J:\mathbb R_+\rightarrow \mathbb R_+$ is a $C^1$-diffeomorphism.

Now, regarding the $C^\infty$ smoothness of $J$, we first notice, that we cannot apply Leibniz's rule of differentiation directly to \eqref{Jprime}, since this will lead to nonintegrable singularities $\vert t-t_\pm(E)\vert^{-n/2}$ with $n>1$. Instead, we use integration by parts and for this purpose, it is easiest to divide $J$ into the sum of two terms:
\begin{align*}
 J_-(E) = \int_{t_-(E)}^0 \sqrt{U(t,E)}dt,\quad J_+(E) = \int_0^{t_+(E)} \sqrt{U(t,E)} dt. 
\end{align*}
 For simplicity we focus on $J_-$; the analysis of $J_+$ is identical. 

%
%
%
%
Since $t=t_-(E)$ is a simple root of $U(t,E)$ for all $E>0$, we have 
\begin{align*}
 \sqrt{U(t,E)} = \sqrt{a(E)(t-t_-(E))} \widetilde U_-(t,E), 
\end{align*}
with $a(E)>0$, and $\widetilde U_-(t,E)>0$ for $t\in [t_-(E),0]$, both $C^\infty$-smooth. We therefore write $J_-$ as $\sqrt{a(E)}\widetilde J_-$ with
\begin{align*}
 \widetilde J_-(E) = \int_{t_-(E)}^0 \sqrt{t-t_-(E)} \widetilde U_-(t,E) dt,
\end{align*}
for $t\in (t_-(E),t_+(E))$. By showing that $\widetilde J_-$ is $C^\infty$-smooth it follows that $J_-$ is $C^\infty$-smooth. We omit the tildes henceforth.

We claim:
\begin{equation}\eqlab{F1n}
\begin{aligned}
 J_-(E)& = \sum_{k=1}^{n} (-1)^{k-1} \frac{2^k}{(2k+1)!!} (-t_-(E))^{\frac{2k+1}{2}} \frac{\partial^{k-1}}{\partial t^{k-1}} \widetilde U_-(0,E)\\
 &+(-1)^n \frac{2^n}{(2n+1)!!} \int_{t_-(E)}^0 (t-t_-(E))^{\frac{2n+1}{2}}\frac{\partial^n}{\partial t^n} \widetilde U_-(t,E)dt.
\end{aligned}
\end{equation}
for every $n\in \mathbb N$. 
The result is true for $n=0$. We therefore prove \eqref{F1n} by induction, assuming that it holds for $n-1$ such that
\begin{align*}
J_-(E)& = \sum_{k=1}^{n-1} (-1)^{k-1} \frac{2^k}{(2k+1)!!} (-t_-(E))^{\frac{2k+1}{2}} \frac{\partial^{k-1}}{\partial t^{k-1}} \widetilde U_-(0,E)\\
 &+(-1)^{n-1} \frac{2^{n-1}}{(2n-1)!!} \int_{t_-(E)}^0 (t-t_-(E))^{\frac{2n-1}{2}}\frac{\partial^{n-1}}{\partial t^{n-1}} \widetilde U_-(t,E)dt.
\end{align*}
We now use integration by parts, writing 
\begin{align*}
(t-t_-(E))^{\frac{2n-1}{2}}= \frac{2}{2n+1} \frac{d}{dt} (t-t_-(E))^{\frac{2n+1}{2}}.
\end{align*}
This gives 
\begin{align*}
J_-(E)& = \sum_{k=1}^{n-1} (-1)^{\rsp{k-1}} \frac{2^k}{(2k+1)!!} (-t_-(E))^{\frac{2k+1}{2}} \frac{\partial^\rsp{k-1}}{\partial t^\rsp{k-1}} \widetilde U_-(0,E)\\
 &+(-1)^{n-1} \frac{2^{n}}{(2n+1)!!} (-t_-(E))^{\frac{2n+1}{2}} \frac{\partial^{n-1}}{\partial t^{n-1}} \widetilde U_-(0,E)\\
 &+(-1)^n \frac{2^{n}}{(2n+1)!!} \int_{t_-(E)}^0 (t-t_-(E))^{\frac{2n+1}{2}}\frac{\partial^{n}}{\partial t^{n}} \widetilde U_-(t,E)dt\\
 &=\sum_{k=1}^{n} (-1)^{\rsp{k-1}} \frac{2^k}{(2k+1)!!} (-t_-(E))^{\frac{2k+1}{2}} \frac{\partial^\rsp{k-1}}{\partial t^\rsp{k-1}} \widetilde U_-(0,E)\\
 &+(-1)^n \frac{2^{n}}{(2n+1)!!} \int_{t_-(E)}^0 (t-t_-(E))^{\frac{2n+1}{2}}\frac{\partial^{n}}{\partial t^{n}} \widetilde U_-(t,E)dt,
\end{align*}
as desired. 

Fix any $N\in \mathbb N$. To see that $J_-$ is $C^N$-smooth with respect to $E>0$, consider \eqref{F1n} with $n=N$.  Then the finite sum is a smooth function of $E>0$. Moreover, we can apply the Leibniz rule of differentiation to the final integral $N$ number of times. Hence $J_-$ is $C^\infty$-smooth and since $J_+$ can be handled in the same way, we conclude that $J$ itself is $C^\infty$-smooth with respect to $E>0$.

We now finally turn to the proof of item \ref{Jprop3}. By assumption \ref{Vquadratic}, we write 
\begin{align*}
 V(t) = t^2 \overline V(t),
\end{align*}
with $\overline V(0)=1$, $\overline V(t)>0$ by assumption \ref{Vquadratic}. Seeing that $t_\pm (E)$ are roots of $U(t,E) = E-V(t)$, it follows that 
\begin{align*}
 t_\pm (E) = \sqrt{E} \overline t_\pm (\sqrt{E}),
\end{align*}
with $\overline t_\pm$ being roots of
\begin{align}
 E^{-1} U(\sqrt{E} \overline t,E) = 1-\overline t^2 \overline V(\sqrt{E}\overline t) = 0.\eqlab{overlinetEqn}
\end{align}
Here $\overline t_\pm:\mathbb R\rightarrow \mathbb R$ are both smooth, specifically $\overline t_\pm(0)=\pm 1$ and 
\begin{align}
 \overline t_\pm (-\sqrt{E}) = -\overline t_\mp(\sqrt{E}),\eqlab{this}
\end{align}
due to the invariance of \eqref{overlinetEqn} with respect to $(\overline t, E)\mapsto (-\overline t,-E)$.
We then write $J(E)$ as 
\begin{align*}
 J(E) = E\overline J(\sqrt{E}),\quad \overline J(\sqrt{E}):= \int_{\overline t_-(\sqrt{E})}^{\overline t_+(\sqrt{E})} \sqrt{1-\overline t^2 \overline V(\sqrt{E} \overline t)} d\overline t. 
\end{align*}
Clearly, $\overline J(0)=\frac12$.
We can then prove that $\overline J$ is smooth in a neighborhood of $\sqrt{E}=0$ in the same way as above for $J$. Then, upon using \eqref{this}, a simple change of variables show that $\overline J$ is an even function
\begin{align*}
 \overline J(-\sqrt{E})=\overline J(\sqrt{E}).
\end{align*}
It is therefore in fact smooth with respect to $E$ and item \ref{Jprop3} of \lemmaref{Jprop} follows. 

%
\end{document}